\documentclass[11pt]{article}
\usepackage[utf8]{inputenc}

\usepackage{amsmath, amssymb, amsthm, amsfonts}
\usepackage{mathtools}
\usepackage{thm-restate}

\usepackage[dvipsnames,svgnames,x11names]{xcolor}
\definecolor{ForestGreen}{rgb}{0.1333,0.5451,0.1333}
\usepackage[colorlinks=true,linkcolor=ForestGreen,citecolor=blue]{hyperref}

\usepackage[margin=1in]{geometry}
\usepackage{graphics}
\usepackage{pifont}
\usepackage{tikz}
\usepackage{bbm}
\usepackage[T1]{fontenc}

\usetikzlibrary{arrows.meta}
\usepackage{environ}
\usepackage{framed}
\usepackage{url}
\usepackage{algorithm}
\usepackage[algo2e,linesnumbered,lined,ruled,boxed]{algorithm2e}
\usepackage{algpseudocode}
\usepackage[noabbrev,capitalize,nameinlink]{cleveref}
\crefname{equation}{}{}
\usepackage[labelfont=bf]{caption}
\usepackage{cite}
\usepackage{framed}
\usepackage[framemethod=tikz]{mdframed}
\usepackage{appendix}
\usepackage{graphicx}
\usepackage[textsize=tiny]{todonotes}
\usepackage{tcolorbox}
\allowdisplaybreaks[1]

\newcommand\remove[1]{}

\newtheorem{lemma}{Lemma}[section]
\newtheorem{theorem}{Theorem}
\newtheorem*{lemma*}{Lemma}
\newtheorem{corollary}[lemma]{Corollary}
\newtheorem*{corollary*}{Corollary}

\newtheorem*{remark}{Remark}

\theoremstyle{definition}

\newtheorem*{theorem*}{Theorem}
\newtheorem{definition}[lemma]{Definition}

\newtheorem*{rem*}{Remark}

\newcommand{\mc}{\mathcal}
\newcommand{\eps}{\varepsilon}
\newcommand{\R}{\mathbb{R}}

\newcommand{\E}{\mathop{\mathbb{E}}}

\newcommand{\norm}[1]{\left\lVert#1\right\rVert}

\crefname{algocf}{Algorithm}{Algorithms}
\renewcommand{\l}{\langle}
\renewcommand{\r}{\rangle}
\newcommand{\supp}{\mathsf{supp}}
\newcommand{\Z}{\mathbb{Z}}
\newcommand{\D}{\mathbb{D}}

\newcommand{\bbC}{\mathbb{C}}

\newcommand{\A}{\Sigma}
\newcommand{\B}{\Gamma}
\newcommand{\C}{\Phi}
\renewcommand{\a}{\sigma}

\renewcommand{\bar}{\overline}
\renewcommand{\hat}{\widehat}

\newcommand{\val}{\mathsf{swap}}
\newcommand{\wt}{\widetilde}

\renewcommand{\bar}{\overline}

\renewcommand{\deg}{\mathsf{deg}}

\renewcommand{\Re}{\mathsf{Re}}
\renewcommand{\Im}{\mathsf{Im}}

\newcommand{\lr}{\leftrightarrow}

\newcommand{\boxx}{\mathsf{box}}
\newcommand{\F}{\mathbb{F}}
\newcommand{\dtv}{\mathrm{d}_{\mathrm{TV}}}
\newcommand{\DP}{\mathsf{DP}}
\newcommand{\stab}{\mathsf{Stab}}
\newcommand{\cons}{\mathsf{Cons}}
\newcommand{\dkl}{\mathrm{d}_{\mathrm{KL}}}
\renewcommand{\ge}{\geqslant}
\renewcommand{\le}{\leqslant}
\renewcommand{\geq}{\geqslant}
\renewcommand{\leq}{\leqslant}
\newcommand{\lst}{\mathsf{List}}
\newcommand{\slst}{\mathsf{ShortList}}

\begin{document}

\title{On Approximability of Satisfiable $k$-CSPs: VI}

\author{Amey Bhangale\thanks{Department of Computer Science and Engineering, University of California, Riverside. Supported by the Hellman Fellowship award.}
	\and
	Subhash Khot\thanks{Department of Computer Science, Courant Institute of Mathematical Sciences, New York University. Supported by
		the NSF Award CCF-1422159, NSF CCF award 2130816, and the Simons Investigator Award.}
	\and
    Yang P. Liu\thanks{School of Mathematics, Institute for Advanced Study, Princeton, NJ. This material is based upon work supported by the National Science Foundation 
under Grant No. DMS-1926686}
    \and 
	Dor Minzer\thanks{Department of Mathematics, Massachusetts Institute of Technology. Supported by NSF CCF award 2227876 and NSF CAREER award 2239160.}}
\date{\vspace{-5ex}}
\clearpage\maketitle

\begin{abstract}
We prove local and global inverse theorems for general $3$-wise correlations over pairwise-connected distributions.
Let $\mu$ be a distribution over $\A \times \B \times \C$ such that the supports of $\mu_{xy}$, $\mu_{xz}$, and $\mu_{yz}$ are all connected, and let $f: \A^n \to \bbC$, $g: \B^n \to \bbC$, $h: \C^n \to \bbC$ be $1$-bounded functions satisfying \[ \left|\E_{(x,y,z) \sim \mu^{\otimes n}}[f(x)g(y)h(z)]\right| \ge \eps. \] 
In this setting, our local inverse theorem asserts that there is $\delta := \exp(-\eps^{-O_{\mu}(1)})$ such that with probability at least $\delta$, a random restriction of $f$ down to $\delta n$ coordinates $\delta$-correlates to a product function. To get a global inverse theorem, we prove a restriction inverse theorem for general product functions, stating that if a random restriction of $f$ down to $\delta n$ coordinates is $\delta$-correlated with a product function with probability at least $\delta$, then $f$ is $2^{-\textsf{poly}(\log(1/\delta))}$-correlated with a function of the form $L\cdot P$, where $L$ is a function of degree $\textsf{poly}(1/\delta)$,  $\|L\|_2\leq 1$, and $P$ is a product function.

We show applications to property testing and to additive combinatorics. 
In particular, we show the following result via a density increment argument.
Let $\A$ be a finite set and $S \subseteq \A \times \A \times \A$ such that: (1) $(x, x, x) \in S$ for all $x \in S$, and (2) the supports of $S_{xy}$, $S_{xz}$, and $S_{yz}$ are all connected. Then, any set $A \subseteq \A^n$ with $|\A|^{-n}|A| \ge \Omega((\log \log \log n)^{-c})$ contains $x, y, z \in A$, not all equal, such that $(x_i,y_i,z_i) \in S$ for all $i$. This gives the first reasonable bounds for the restricted 3-AP problem over finite fields.
\end{abstract}

\pagenumbering{gobble}

\newpage

\setcounter{tocdepth}{2}
\tableofcontents

\normalsize
\pagebreak
\pagenumbering{arabic}

\section{Introduction}\label{sec:intro}
A recent line of research~\cite{BKM1,BKM2,BKM3,BKM4,BKM5} about the approximability of satisfiable constraint satisfaction problems (CSPs in short) identified the following analytical problem as a central component. Suppose that $\Sigma$, $\Gamma$ and $\Phi$ are finite alphabets (thought of as being of constant size), and let $\mu$ be a distribution over $\Sigma\times\Gamma\times\Phi$ in which the probability of each atom is at least $\Omega(1)$. What triplets of $1$-bounded functions 
$f\colon \Sigma^n\to\mathbb{C}$, 
$g\colon \Gamma^n\to\mathbb{C}$,
$h\colon \Phi^n\to\mathbb{C}$ can achieve a significant $3$-wise correlation with respect to $\mu$, i.e., satisfy that
\begin{equation}\label{eq:main_analytical}
\left|\E_{(x,y,z) \sim \mu^{\otimes n}}[f(x)g(y)h(z)]\right|\geq \eps?
\end{equation}
To get a meaningful answer, one must make some assumptions about the distribution $\mu$; otherwise, $\mu$ could be supported only on inputs of the form $(x,x,x)$, in which case one could take $f$ to be any function of large $2$-norm, $g(y) = \overline{f(y)}$ and $h=1$. Here and throughout, we will restrict our discussion to the class of pairwise-connected distributions, defined as follows.
\begin{definition}
    For finite alphabets $\Sigma,\Gamma,\Phi$, a distribution $\mu$ over $\Sigma\times\Gamma\times\Phi$ is called pairwise-connected if
    the bipartite graphs $(\Sigma\cup \Gamma, \textsf{supp}(\mu_{x,y}))$, 
    $(\Sigma\cup \Phi, \textsf{supp}(\mu_{x,z}))$,
    $(\Gamma\cup \Phi, \textsf{supp}(\mu_{y,z}))$ are all connected.
\end{definition}

Earlier works have considered smaller classes of distributions. In~\cite{Mossel10}, Mossel shows that if $\mu$ is \emph{connected}, then any triplets of functions satisfying~\eqref{eq:main_analytical} must be correlated with low-degree functions. Here, a distribution is called connected if the graph whose vertices are $\textsf{supp}(\mu)$, and whose edges are between points in $\textsf{supp}(\mu)$ that differ in exactly one coordinate, is connected. In~\cite{BKM1}, the authors consider the notion of \emph{Abelian embeddings} and speculate that solutions to~\eqref{eq:main_analytical} are related to Abelian embeddings of $\mu$.
\begin{definition}
    We say $S\subseteq \Sigma\times\Gamma\times\Phi$ admits an Abelian embedding if there are Abelian group $(G,+)$ and maps $\sigma\colon \Sigma\to G$, $\gamma\colon \Gamma\to G$, $\phi\colon \Phi\to G$ not all constant, such that $\sigma(x)+\gamma(y)+\phi(z) = 0$ for all $(x,y,z)\in S$.  A distribution $\mu$ over $\Sigma\times\Gamma\times\Phi$ is said to admit an Abelian embedding if $\textsf{supp}(\mu)$ admits an Abelian embedding.
\end{definition}
Prior works~\cite{BKM2,BKM4} considered the class of distributions that do not admit any Abelian embedding (which can easily be seen to be larger than the class of connected distributions), and the more general class of distributions that do not admit embeddings to $(\mathbb{Z},+)$. In the former case, it was proved that triplets of functions satisfying~\eqref{eq:main_analytical} must be correlated with  low-degree functions. 
In the latter case, it was proved that triplets of functions satisfying~\eqref{eq:main_analytical} must be correlated with a functions of the form $L\cdot \chi\circ\sigma^{\otimes n}$, 
$L'\cdot \chi'\circ\gamma^{\otimes n}$,
$L''\cdot \chi''\circ\phi^{\otimes n}$, where $L,L',L''$ are low-degree functions with bounded $2$-norm, $\chi,\chi',\chi''\in \widehat{G}^{n}$ are characters, and $\sigma,\gamma,\phi$ is an Abelian embedding of $\mu$ into $G$, where $(G,+)$ is some finite Abelian group. 
\vspace{-1ex}
\paragraph{The issue of $(\mathbb{Z},+)$-embeddings:} the arguments in~\cite{BKM2,BKM4} rely on the fact that if $\mu$ does not admit any embedding into $(\mathbb{Z},+)$, then it can only admit embeddings into Abelian groups of size $O_{\mu}(1)$, of which there are $O_{\mu}(1)$ many (up to isomorphism). In particular, one can take a product of all such groups and be sure that the resulting Abelian group will ``contain'' within it all embeddings of $\mu$. 
This fact completely breaks once $\mu$ does have embeddings into $(\mathbb{Z},+)$. Indeed, such embedding automatically gives rise to embeddings into  $(\mathbb{Z}_p,+)$ for any prime $p$. 
Despite of that, even in the presence of a $(\mathbb{Z},+)$ embedding, it is not clear if there are solutions to~\eqref{eq:main_analytical} that do not fall into the ``low-degree times embedding function'' mold.

\subsection{Main Results}
\subsubsection{The Global Inverse Theorem}
Our main result is an inverse theorem even in the presence of $(\mathbb{Z},+)$ embeddings. In this case, embedding functions are quite arbitrary, and their main notable feature is that they are \emph{product functions}, in the sense that they are product of functions each depending only on one of the input's coordinates. With this in mind, our main result is the following statement. 
\begin{theorem}\label{thm:main_global}
Let $\alpha>0$, let $\Sigma,\Gamma,\Phi$ be finite alphabets, and let $\mu$ be a pairwise-connected distribution over $\Sigma\times\Gamma\times\Phi$ in which the probability of each atom is at least $\alpha$. Then for every $\eps>0$ there exist $d\in\mathbb{N}$ and $\delta>0$ such that the following holds.
If $f\colon \Sigma^n\to\mathbb{C}$, 
$g\colon \Gamma^n\to\mathbb{C}$, 
$h\colon \Phi^n\to\mathbb{C}$ 
are $1$-bounded functions such that 
\[ \left|\E_{(x,y,z) \sim \mu^{\otimes n}}[f(x)g(y)h(z)] \right| \ge \eps, \]
then there exist $L\colon \Sigma^n\to\mathbb{C}$ of degree at most $d$ and $\|L\|_2\leq 1$, and 
a product function $P\colon \Sigma^n\to\mathbb{C}$ of the form $P(x) = \prod\limits_{i=1}^{n}P_i(x_i)$ where $|P_i(x_i)| = 1$ for all $i$ and $x$, such that 
$|\langle f, L\cdot P\rangle|\geq \delta$. 
Quantitatively, we have that $d = \exp((1/\eps)^{O_{\alpha}(1)})$, 
and $\delta = \exp(-\exp((1/\eps)^{O_{\alpha}(1)}))$.
\end{theorem}

\subsubsection{The Local Inverse Theorem}
The main new component in the proof of~\Cref{thm:main_global} is a similar looking local inverse theorem. By that, we mean that given functions $f$, $g$ and $h$ satisfying~\eqref{eq:main_analytical}, we prove a similar structural result in a local (and random) part of the domain, given by \emph{random restrictions}. For a function $f\colon \Sigma^n\to\mathbb{C}$, a subset $I\subseteq [n]$ of the coordinates, and a setting $z\in \Sigma^{I}$, the \emph{restricted function} 
$f_{I\rightarrow z}$ is a function $f_{I\rightarrow z}\colon \Sigma^{[n]\setminus I} \to \mathbb{C}$ defined as $f_{I\rightarrow z}(y) = f(x_I = y, x_{\overline{I}} = z)$.

In a random restriction, either the set $I$ is sampled randomly, the fixing $z\in \Sigma^I$ is sampled randomly, or both. Below, we use the notation $I \sim_{1-\alpha} [n]$ to mean that $I$ is a random subset of $[n]$ in which each $i \in [n]$ is included in $I$ with probability $1-\alpha$, and $z \sim \nu^I$ means that $z \in \A^I$ is such that each $z_i$ is independently distributed according to $\nu$.
\begin{definition}[Random restriction]
Let $\mu$ be a distribution over $\A^n$ and write $\mu = (1-\alpha)\nu + \alpha \mu'$ for distributions $\nu, \mu'$. Then for a function $f: \A^n \to \bbC$, $I \sim_{1-\alpha} [n]$, and $z \sim \nu^I$, we define the random restriction $f_{I\to z}: (\A^{[n] \setminus I},{\mu'}^{[n]\setminus I}) \to \bbC$ as $f_{I \to z}(x) := f(x, z)$.
\end{definition}

Our local inverse theorem asserts that if $1$-bounded functions $f, g, h$ satisfy~\eqref{eq:main_analytical}, then with noticeable probability, a random restriction of $f$ (and similarly of $g$, $h$) correlates to a $1$-bounded product function. More precisely:
\begin{theorem}
\label{thm:main}
Let $\alpha>0$, let $\Sigma,\Gamma,\Phi$ be finite alphabets, and let $\mu$ be a pairwise-connected distribution over $\Sigma\times\Gamma\times\Phi$ in which the probability of each atom is at least $\alpha$. 
Then for every $\eps>0$ there exists $\delta>0$ such that if $f: \A^n \to \bbC$, $g: \B^n \to \bbC$, $h: \C^n \to \bbC$ are $1$-bounded functions satisfying that
\[ \left|\E_{(x,y,z) \sim \mu^{\otimes n}}[f(x)g(y)h(z)] \right| \ge \eps.,
\] 
then, writing $\mu = (1-\delta)\nu + \delta U$ where $U$ is uniform over $\A$, we have that:
\[ \Pr_{I \sim_{1-\delta} [n], z \sim \nu^I}\left[\exists \{P_i: \A \to \bbC, \|P_i\|_\infty \le 1\}_{i \in \bar{I}} \enspace \text{ with } \enspace \Big|\E_{x \sim \A^{\bar{I}}}\Big[f_{I\to z}(x) \prod_{i \in \bar{I}} P_i(x_i) \Big]\Big| \ge \delta \right] \ge \delta. \]
Quantitatively, $\delta(\eps) \ge \exp(-\eps^{-O_{\alpha}(1)})$.
\end{theorem}
\begin{remark}
A few remarks are in order. First, as far as global inverse theorems are concerned, the low-degree part $L$ is necessary. 
Second, the correlation parameter $\delta$ in Theorem~\ref{thm:main} is exponentially better than the one in Theorem~\ref{thm:main_global}, and thus it gives better quantitative bounds in some applications. Third, we do note know whether the exponential dependency in Theorem~\ref{thm:main} or the double exponential dependency in Theorem~\ref{thm:main_global} is necessary. As far as we know, both results may be true with $\delta = \eps^{O_{\alpha}(1)}$. For instance, in the case no $(\mathbb{Z},+)$ embeddings exist, most of the arguments in~\cite{BKM4} give quasi-polynomial type dependency between the two parameters.
\end{remark}


\subsection{Applications}
In this section we discuss a few applications 
of our main results.
\subsubsection{Restricted 3-APs}
\label{subsec:intro3ap}
Let $p\geq 3$ be a prime thought of as constant. 
A restricted 3-AP in $\mathbb{F}_p^n$ is an arithmetic progression $x,x+a,x+2a$ where 
$x\in\mathbb{F}_p^n$ and $a\in\{0,1\}^n\setminus\{\vec{0}\}$.
The restricted 3-AP problem asks what is the maximum density of a subset $A \subseteq \F_p^n$ which is free of all restricted 3-AP's, i.e., that does not contain any triplet of the form $x, x+a, x+2a \in A$ for some $x \in \F_p^n$ and $a \in \{0, 1\}^n \setminus \{\vec{0}\}$. The restricted 3-AP problem can be 
seen as a variant of Roth's well known result 
in the finite field setting~\cite{roth1953certain,meshulam1995subsets}, which was highlighted by Green~\cite{Green}. 

It follows from the density Hales-Jewett theorem
that the density of a set $A$ that doesn't contain any restricted 3-APs is vanishing with 
$n$, and using the quantitative bounds from~\cite{polymath2012new} one has 
$\frac{|A|}{p^n}\leq O\big(\frac{1}{\sqrt{\log^* n}}\big)$.
Below, we state a result that (among other things) gives the first reasonable bounds (i.e., finite number of iterated logarithms) for a set $A\subseteq \mathbb{F}_p^n$ that doesn't contain any restricted 3-APs.
\begin{restatable}{theorem}{threeap}
\label{thm:3ap}
Let $\A$ be a finite set and let $S \subseteq \A \times \A \times \A$ be a non-empty set satisfying that:
\begin{enumerate}
\item $(x,x,x) \in S$ for all $x \in \A$, and
\item $\supp(S_{xy})$, $\supp(S_{xz})$, $\supp(S_{yz})$ are all connected.
\end{enumerate}
Then there are constants $c_{\A}, C_{\A} > 0$ such that for any $A \subseteq \A^n$ with 
$\frac{|A|}{|\Sigma|^n} \ge \frac{C_{\A}}{(\log\log\log n)^{c_{\A}}}$, there are $x,y,z \in A$, not all equal, such that $(x_i,y_i,z_i) \in S$ for all $i \in [n]$. 
\end{restatable}
Thus, taking $\Sigma = \mathbb{F}_p$ and 
$S = \{(x,x+a,x+2a)~|~x\in\mathbb{F}_p, a\in\{0,1\}\}$, one sees that $S$ satisfies the conditions of~\Cref{thm:3ap}, and the 
conclusion regarding the density of 3-AP free sets follows. In fact, we can take an even sparser $S$ and conclude a similar result. For example, for $p=3$ we can take 
\[
S = \{(0,0,0),(1,1,1),(2,2,2),(0,1,2),(1,2,0)\}, \]
and ensure that a set $A$ with density exceeding $\frac{C_{\A}}{(\log\log\log n)^{c_{\A}}}$ must contain a triplet $x,y,z$ not all equal such that $(x_i,y_i,z_i)\in S$ for all $i$. In comparison, the standard restricted 3-AP problem also allows for $(x_i,y_i,z_i) = (2,0,1)$.

\cref{thm:3ap} and the consequence for restricted $3$-APs free sets improve upon a recent result of~\cite{BKM3ap}, who proved an analogous statement in the case the common difference $a$ is allowed to be in $\{0,1,2\}^n\setminus\{\vec{0}\}$. Our proof of~\cref{thm:3ap} proceeds by combining \cref{thm:main} with a density increment argument. This derivation is somewhat similar to the argument in~\cite{BKM3ap}, however the approach we take here is simpler and applies in more generality.

\subsubsection{Direct Sum Testing in the Low Soundness Regime}
Our next application is to the direct sum testing problem~\cite{david2015direct,dinur2019direct,westover2024new}. More specifically, consider the so-called diamond test, introduced in~\cite{dinur2019direct} and recently analyzed in~\cite{westover2024new}, which proceeds as follows. Given a function $f\colon \Sigma^n\to\mathbb{F}_p$, sample $a,b\in \Sigma^n$, and 
$x\in \{0,1\}^n$ uniformly and independently. Denoting by $\phi_x(a,b)$ the point in $\Sigma^n$ where $\phi_x(a,b)_i = a_i$ for $i$'s such that $x_i=1$ and $\phi_x(a,b) = b_i$ for $i$'s such that $x_i=0$, the tester reads the values of $f$ at the points $a$, $\phi_x(a,b)$, $\phi_x(b,a)$ and $b$, and accepts if and only if 
\[
f(a) - f(\phi_x(a,b)) - f(\phi_x(b,a)) + f(b) = 0.
\]
It is clear that any function $f$ which is a direct sum, i.e.~any function of the form 
$f(a) = \sum\limits_{i=1}^{n}f_i(a_i)$, passes the test with probability $1$. For $p=2$, the paper~\cite{dinur2019direct} proposed the diamond tester as a natural $4$-query test of direct sum-ness, and conjectured that in works in the so-called $99\%$ regime. This conjecture was recently confirmed in~\cite{westover2024new}, wherein the authors showed that if $f$ passes the test with probability $1-\eps$, then there is a direct sum function $g$ such that
\[
\Pr_{x\in \Sigma^n}\left[f(x)\neq g(x)\right]\leq O_{|\Sigma|}(\eps).
\]
Using the techniques underlying~\Cref{thm:main_global}, we are able to analyze this test in the small soundness regime. More precisely, we show the following result:
\begin{theorem}\label{thm:direct_sum}
    For all finite alphabets $\Sigma$ and $\eps>0$, there are $d\in\mathbb{N}$ and $\eps'>0$ such that the following holds.
    Suppose that $f\colon \Sigma^n\to\mathbb{F}_p$
    passes the diamond test with probability at least $\frac{1}{p}+\eps$. Then there exists $\alpha\in\mathbb{F}_p\setminus\{0\}$, a function $L\colon \Sigma^n\to\mathbb{C}$ with $\|L\|_2\leq 1$ and degree at most $d$ and 
    $P\colon \Sigma^n\to\mathbb{C}$ a product function $P(x) = \prod\limits_{i=1}^{n}P_i(x_i)$ where $|P_i(x_i)| = 1$ for all $i$ and $x$, such that
    \[
    \left|\E_x\left[\omega_p^{\alpha f(x)}L(x)P(x)\right]\right|\geq \eps',
    \]
    where $\omega_p$ is a primitive root of unity of order $p$. Quantitatively, the bounds on $d$ and $\eps'$ are the same as in Theorem~\ref{thm:main_global}.
\end{theorem}
We remark that the complicated form of correlation as in the statement of~\Cref{thm:direct_sum} is necessary. 
In the $99\%$ regime, i.e.~in the case that the test passes with probability $1-\eps$, our argument allows us to recover a similar to~\cref{thm:direct_sum} without the low-degree part $L$ and with $\eps' = 1-O_p(\sqrt{\eps})$.

\subsection{The Swap Norm}
\label{subsec:gowers}
Underlying the proof of our inverse theorems is the \emph{swap norm}, which is a new norm that can be viewed as an analogue of the $U^2$-Gowers' uniformity norm. To define the swap norm, we introduce the following convenient notation.
\begin{definition}
\label{def:swap}
Given $x, y \in \A^n$ we write $(x',y') \sim (x\lr y)$ to denote the following distribution over $\Sigma^n\times \Sigma^n$. We sample $(x',y')$ to satisfy $(x_i',y_i') = (x_i,y_i)$ or $(y_i,x_i)$ uniformly, independently for each coordinate $i$.
\end{definition}
In words, given $x,y\in \A^n$, the values of coordinates in the distribution $(x',y') \sim (x\lr y)$ is either the same as in $(x,y)$, or else it is swapped, and the choice in each coordinate is made independently. With this notation, we define the swap inner product and the swap norm.

\begin{definition}[Swap inner product and swap norm]
\label{def:swapnorm}
For functions $f_1, f_2, f_3, f_4: \A^n \to \mathbb{C}$, define
\[ \val(f_1,f_2,f_3,f_4) := 
\E_{\substack{x \sim \A^n, y \sim \A^n \\ (x', y') \sim (x \lr y)}} \left[f_1(x)f_2(y)\bar{f_3(x')f_4(y')}\right].
\]
When $f_1=f_2=f_3=f_4=f$, we abbreviate this notation and write $\val(f) := \val(f,f,f,f)$. 
\end{definition}

The form $\val(f_1,f_2,f_3,f_4)$ 
is easily seen to be a multi-linear form, 
and we show that it satisfies several properties that are reminiscent of Gowers' norms~\cite{Gowers01}; see \cref{sec:properties} for details. Among other things, we prove that $\val(f)^{1/4}$ is a norm over functions $f$ (see~\Cref{cor:norm}).

To get some intuition as to the relevance of the swap norm to product functions, we note that if 
$f(x) = \prod\limits_{i=1}^{n}P_i(x_i)$ is a product function with $\norm{f}_2=1$, then $\val(f) = \norm{f}_2^2=1$. More generally, a quick Cauchy-Schwarz argument shows that if $f$ is $\eps$-correlated with a product function $P$ with $\norm{P}_2=1$, then $\val(f)\geq \eps^4$.\footnote{This is similar to the way that if $f$ is correlated with a character then it has noticeable $U^2$ uniformity norm, and if it is correlated with a ``degree $s$'' polynomial then it has noticeable $U^{s+1}$ uniformity norm.}

Our local inverse theorem for the swap norm asserts that this is essentially the only possible way to have a large swap norm, at least under random restrictions. More precisely, we show:
\begin{restatable}{theorem}{swap}
\label{thm:swap}
If $f: \A^n \to \bbC$ is a $B$-bounded function satisfying that $\|f\|_2 \le 1$, $\val(f)^{1/4} \ge \eps$, then there is a constant $\delta(\eps, B) > 0$ such that:
\[ \Pr_{I\sim_{1-\delta(\eps, B)} [n], z \sim \A^I}\left[\exists \{P_i: \A \to \bbC, \|P_i\|_2 \le 1\}_{i \in \bar{I}} \enspace \text{ with } \enspace \Big|\Big\l f_{I\to z}, \prod_{i\in \bar{I}}P_i \Big\r \Big| \ge \delta(\eps, B) \right] \ge \delta(\eps, B). \]
Quantitatively, $\delta(\eps, B) \ge (\eps/B)^{O(\eps^{-O(1)})}$.
\end{restatable}
\vspace{-1ex}
\paragraph{Deducing~\Cref{thm:main} from~\Cref{thm:swap}:}
to prove \cref{thm:main}, we show that for any triplets of functions $f$, $g$ and $h$ satisfying~\eqref{eq:main_analytical} with respect to some distribution $\mu$ which is pairwise-connected, it must be the case that 
the swap norm of $f$ is noticeable, at least in expectation over appropriately chosen random restrictions. This deduction is shown via a sequence of applications of the Cauchy-Schwarz inequality (which we call the \emph{path trick}, see \cref{def:pathtrick}). This gets us that a form similar to the one in the definition of the swap inner product, except that the distribution over $(x,y)$ is not necessarily uniform, is noticeable. We then apply random restrictions to modify the distribution of $(x,y)$ to be uniform. This step closely mirrors the situation with Szemer\'{e}di's theorem: if $f$ has nontrivial correlation over the distribution supported on $k$-APs, then by performing a sequence of Cauchy-Schwarz manipulations, we conclude that $f$ has a large $U^{k-1}$ norm. 

\vspace{-1ex}
\paragraph{The proof of~\Cref{thm:swap}:} the bulk
of our effort goes into establishing~\Cref{thm:swap}, and our argument proceeds by induction on $\eps$. One challenging aspect of the proof is that 
for the conclusion to be true, one must consider random restrictions. In other words, there are functions $f$ that have a noticeable swap norm that have $o(1)$ correlation with all product functions.\footnote{For example, one can take $\Sigma = \{-1,1\}$ and a random signed quadratic function, namely taking a set of random independent signs $\{\alpha_{i,j}\}_{1\leq i<j\leq n}$ and defining $f(x) = \textsf{sign}(\sum\limits_{i<j}\alpha_{i,j}x_ix_j)$.} This means that our inductive proof strategy must incorporate within it random restrictions all the way, and we defer further discussion to~\Cref{sec:overview}.

\subsection{Subsequent and Future Works}\label{subsec:related}
\paragraph{Subsequent works:} in~\cite{BKLM7},~\cite{BKLMDHJ3} we use the results proved herein to settle a conjecture from~\cite{BKM1}, as well as give reasonable bounds for the density Hales-Jewett theorem over $[3]^n$. 

More specifically, in~\cite{BKLM7} we prove that~\cite[Hypothesis 1.6]{BKM1} is true. To state this result, we first recall the natural extension of the definition of Abelian embeddings to $k$-ary distributions. We say 
that a distribution $\mu$ over $\Sigma_1\times\ldots\times\Sigma_k$ admits an Abelian embedding if there are Abelian group $(G,+)$ and maps $\sigma_i\colon \Sigma_i\to G$ not all constant, such that $\sum\limits_{i=1}^{k}\sigma_k(x_i) = 0$
for all $(x_1,\ldots,x_k)\in \supp(\mu)$. The main result of~\cite{BKLM7} is the following assertion:
\begin{theorem}\label{thm:proved_hypothesis}
    Let $k\geq 3$ and let $\Sigma_1,\ldots,\Sigma_k$ be finite alphabet. Suppose that $\mu$ is a distribution over $\Sigma_1\times\ldots\times\Sigma_k$ in which the probability of each atom is at least $\alpha$, and further suppose that $\mu$ does not admit any Abelian embedding. If $f_i\colon \Sigma_i^n\to\mathbb{C}$ are $1$-bounded functions for $i=1,\ldots,k$, such that
    \[
    \left|\E_{(x_1,\ldots,x_k)\sim \mu^{\otimes n}}\left[\prod\limits_{i=1}^kf_i(x_i)\right]\right|\geq \eps,
    \]
    then there exists a function $L\colon \Sigma^n\to\mathbb{C}$ of degree $O_{\alpha,\eps}(1)$ and $\|L\|_2\leq 1$ such that 
    $\langle f_1,L\rangle\geq \Omega_{\alpha,\eps}(1)$.
\end{theorem}
The paper~\cite{BKLM7} also establishes extensions of this result which are necessary for the application to the density Hales-Jewett problem, and we refer the reader there for further discussion. 

In~\cite{BKLMDHJ3}, we establish reasonable bounds for the density Hales-Jewett theorem~\cite{HJ,FK,polymath2012new}, which can be thought of as an even more restricted form of the restricted 3-APs problem. 
In this context, a combinatorial line is a triplet $x,y,z\in \{0,1,2\}^n$, not all equal, where for each $i$, 
\[
(x_i,y_i,z_i)\in \{(0,0,0), (1,1,1), (2,2,2), (0,1,2)\}.
\]
The main result of~\cite{BKLMDHJ3} asserts that if $A\subseteq \{0,1,2\}^n$ does not contain any combinatorial line, then $\frac{|A|}{3^n}\leq \frac{1}{(\log\log\log\log n)^c}$, where $c>0$ is an absolute constant.
\vspace{-1ex}
\paragraph{Future works:} in future works, we plan to investigate the algorithmic applications of~\Cref{thm:main} for the class of satisfiable constraint satisfaction problems of arity $3$, extending the result of~\cite{BKM5} to a wider class of predicates. An example for a predicate of interest which belongs to this larger class (and for which the result of~\cite{BKM5} does not apply) is the rainbow coloring of $3$-uniform hypergraph. In this problem, the input is a $3$-uniform hypergraph $\mathcal{H}$ promised to be tripatite, and the goal is to partition the vertices into $3$-parts and maximize the number of hyperedges touching all parts.  

We also plan to investigate whether 
one can extend results as~\Cref{thm:3ap} and establish the stronger counting versions. By that, we mean a statement along the lines of: if $A\subseteq \mathbb{F}_p^n$ is a set of density at least $\eps$, then it contains at least $\delta$ fraction of the restricted $3$-APs, where $\delta = \delta(\eps)>0$.

\subsection{Preliminaries}

\paragraph{General notation.}
Let $[n] = \{1, \dots, n\}$. For a subset $I \subseteq [n]$ we write $\bar{I} = [n] \setminus I$. For a real number $x$ we write $\|x\|_{\R/\Z} := \min_{z \in \Z} |x-z|$. We let $\bbC$ denote complex numbers, and let $\D$ denote the unit disk. We let $\supp(\mu)$ denote the support of a distribution $\mu$.
For a finite set $\A$ we write $x \sim \A^n$ to denote sampling $x$ uniformly randomly from $\A^n$.

Our arguments use singular-value-decompositions heavily (abbreviated SVD henceforth), as per the following definition.
\begin{definition}[$J$-SVD]
\label{def:svd}
For a function $f: \A^n \to \bbC$ we say that its $J$-SVD is a representation $f = \sum_i \lambda_ig_ih_i$ for $g_i: \A^J \to \bbC$ and $h_i: \A^I \to \bbC$, where $\lambda_i$ are singular values and $g_i, h_i$ are singular vectors when $f$ is expressed as a matrix $M \in \bbC^{\A^J \times \A^I}$ for $I = [n] \setminus J$ defined as $M_{x,y} := f(x,y)$.
\end{definition}

\section{Proof Overview}
\label{sec:overview}
In this section we give a proof overview for our inverse theorems. We begin by discussing our local inverse theorem,~\cref{thm:main}.

\subsection{From 3-Wise Correlation to a Large Swap Norm}
\label{subsec:overviewtoswap}
The first step towards proving \cref{thm:main} is the following statement, which reduces the proof of~\Cref{thm:main} to~\Cref{thm:swap}.
\label{subsec:toswap}
\begin{restatable}{theorem}{toswap}
\label{thm:toswap}
Let $\mu$ be a pairwise-connected distribution over $\A \times \B \times \C$ for finite sets $\A, \B, \C$. If $1$-bounded functions $f: \A^n \to \bbC$, $g: \B^n \to \bbC$, $h: \C^n \to \bbC$ satisfy that
\[ \left|\E_{(x,y,z) \sim \mu^{\otimes n}}[f(x)g(y)h(z)] \right| \ge \eps, \]
then there is $\alpha \geq \Omega_{|\Sigma|,|\Gamma|,|\Phi|}(1)$, and distribution $\nu$ over $\A$ satisfying $\mu_x = \alpha U + (1-\alpha)\nu$ for uniform $U$ over $\A$, such that
\[ \E_{\substack{I \sim_{1-\alpha} [n] \\ z \sim \nu^I}}\Big[\val(f_{I\to z}, f_{I\to z}, f_{I\to z}, f_{I\to z})\Big] \ge \eps^{O_{\mu}(1)}. \]
\end{restatable}
The idea towards proving \cref{thm:toswap} is to enlarge the alphabets $\A, \B, \C$ and distribution $\mu$ carefully through a sequence of \emph{path tricks}. Formally, the path trick takes an integer $r$, sets $\A^+ \subseteq \A^{2^{r+1}-1}$ and for all $y_1, z_1, y_2, z_2, \dots, y_{2^r}, z_{2^r}$ such that $(y_i, z_i), (z_i, y_{i+1}) \in \supp(\mu_{yz})$, adds the following to $\supp(\mu^+)$. Let $(x_i, y_i, z_i) \in \supp(\mu)$ for $i = 1, \dots, 2^r$, and $(x_i', y_{i+1}, z_i) \in \supp(\mu)$ for $i = 1, 2, \dots, 2^r - 1$.
Then, \[ (\{x_1, x_1', x_2, \dots, x_{2^r}\}, y_1, z_{2^r}) \in \A^+ \times \B \times \C \] is added to $\supp(\mu^+)$. We can define similar path tricks that enlarge $\B$ and $\C$ respectively.

The guarantees of the path trick are given in \cref{lemma:pathtrick}. Informally, if we apply an $x$-path trick, enlarging $\A$, then $f^+, g, \tilde{h}$ have correlation at least $\eps^{O(1)}$, where $\tilde{h}: \C^n \to \bbC$ is still $1$-bounded but otherwise arbitrary, and $f^+: \A^+ \to \bbC$ is defined for $x = (x_1, x_1', \dots, x_{2^{r-1}}) \in (\A^+)^n$ as $f^+(x) = \prod_{i=1}^{2^{r-1}} f(x_i) \prod_{i=1}^{2^{r-1}-1} \bar{f(x_i')}$. In other words, the values multiply.

Now, perform the following sequence of three path tricks.
\begin{enumerate}
    \item Perform a $y$-path trick that makes $xz$ have full support.
    \item Perform a $z$-path trick that makes $xy$ have full support.
    \item Perform a $x$-path trick for $r = 2$.
\end{enumerate}
This only affects the value of $\eps$ polynomially.
Let the resulting distribution be $\mu^+$ over $\A^+ \times \B^+ \times \C^+$. By inspection, one can check that for all $a, b \in \A$ that there is some $y \in \B^+$ and $z \in \C^+$ such that $((a,a,b),y,z)$ and $((b,a,a),y,z)$ are both in the support of $\mu^+$. Intuitively, this means that the first and third coordinates of $\A^+ \subseteq \A^3$ can be swapped without affecting the correlation. Because $f^+(x) = f(x_1)\bar{f(x_1')}f(x_2)$, this means that swapping each coordinate of $x_1, x_2$ should not affect the value of $f^+$ significantly, and thus $f$ has large swap norm.

\subsection{The 99\%~Regime of Swap Norm Inverse Theorem}
\label{subsec:overview99}
The rest of this overview section is devoted to discussing the proof of~\Cref{thm:swap}, which consists of the bulk of the effort in this paper. 
As mentioned earlier, our argument proceeds by induction on $\eps$, and we first address the base case that $\eps$ is sufficiently close to $1$. In this case random restrictions are not necessary, and we establish the following stronger result:
\begin{restatable}{theorem}{swapnine}\label{thm:swapnine}
If $f: \A^n \to \bbC$ satisfies $\|f\|_2 \le 1$ and $\val(f) \ge 1-c$ for $c\leq 0.05$, then there are $P_1, \dots, P_n: \A \to \bbC$ with $\|P_i\|_2 = 1$ such that $|\l f, P_1\dots P_n\r| \ge 1-O(c)$.
\end{restatable}

To prove Theorem~\ref{thm:swapnine}, we build $P_1, \dots, P_n$ one at a time. Let the $\{1\}$-SVD on $f$ be $f = \sum \lambda_i g_i h_i$, where the largest singular value is $\lambda_1$. We will take $P_1 = g_1$, and note that for any function $P: \A^{[n] \setminus \{1\}} \to \bbC$ we have that $\l f, g_1 P \r = \lambda_1 \l h_1, P \r$. We repeat the argument on $h_1$: perform an SVD, take the largest singular value and so on, until we end up with a product function. To argue that this process works, it suffices to argue that $\val(h_1) \ge (1+\Omega(\delta))\val(f)$, where $\delta = 1-\lambda_1^2$. Towards this end, we expand $\val(f)$ (one can check that the orthogonality of $g$ causes other terms to cancel) and get that
\[
1-c
\leq
\val(f)
\leq
\sum\limits_{i}\lambda_i^4\val(h_i,h_i,h_i,h_i)
+
2\sum\limits_{i<j}\lambda_i^2\lambda_j^2\val(h_i,h_j,h_i,h_j)
\]
We require a non-trivial bound on the off-diagonal entries, and we show that $\val(h_i,h_j,h_i,h_j) \le 2/3$ for orthogonal $h_i, h_j$ (\cref{lemma:valperp}), which holds for $i \neq j$. 
Using this estimate, we can show that $\lambda_1^2$ is close to $1$ and hence $\delta$ is close to $0$. It follows that 
$\val(f)\leq \lambda_1^4\val(h_1) + \frac{4}{3}\delta + O(\delta^2)$, giving 
$\val(h_1)\geq \frac{\val(f) - \frac{4}{3}\delta + O(\delta^2)}{1-2\delta + O(\delta^2)}\geq (1+\Omega(\delta))\val(f)$.

\subsection{Inverse Theorem for the Swap Norm}
\label{subsec:overviewswap}
We now discuss the proof of~\Cref{thm:swap}. At a high level, our goal will be to perform a sequence of random restrictions and SVDs to eventually get to a function $g$ such that: 
\begin{enumerate}
    \item {\bf Correlation for $g$ implies correlation for $f$:} if a product function correlates with $g$, then a product function correlates with $f$ up to a multiplicative factor $\eps^{100}$ loss.
    \item {\bf Increase in swap norm:} $\val(g)^{1/4} \ge \eps + \eps^{100}$.
    \item {\bf Boundedness:} if $f$ is bounded, then $g$ is also bounded.
\end{enumerate}
Once we find such a procedure, we can iterate it until the swap norm of the resulting function becomes close to $1$, which happens in $O(1/\eps^{100})$ iterations by the second item. 
At that point, we use the base case, namely~\Cref{thm:swapnine}, and conclude that the resulting function is correlated with a product function. This result then translates to a correlation of a random restriction of $f$ with a product function, where the correlation parameter is at least $(\eps^{100})^{1/\eps^{100}} =  \exp(-\eps^{-O(1)})$.

The starting point of our argument is that if $f$ has large swap norm, then for a random subset $J \subseteq [n]$, the largest singular value of the $J$-SVD of $f$ is at least $\eps^{O(1)}$ in expectation. This follows by \cref{lemma:box}, which states that the swap norm is the expected \emph{box norm} of $f$ with respect to a random partition $(J, \bar{J})$, and the fact that having a large box norm with respect to some partition implies a large singular value in the corresponding SVD decomposition. 

This suggests the following approach: consider the $J$-SVD, say $f = \sum\limits_{i} \lambda_i g_i h_i$, and try to come up with a candidate function $g$ above based in the functions $g_i$. A natural attempt, for instance would be to take $g$ to be a linear combination of the the singular vectors $g_1, g_2, \dots$, possibly only those corresponding to large singular values. However, we do not know how to argue that this works, even in the case where there is only one large singular value. 
The issue is that there is no clear way to directly relate $\val(g_1)\val(h_1)$ and $\val(f)$.

Instead, we take the following iterative approach. Let $\lambda gh$ be the first term of the $J$-SVD of $f$, and write $\Delta = f - \lambda gh$ -- note that $\|\Delta\|_2^2 = \norm{f}_2^2 - \lambda^2 \le 1 - \eps^{O(1)}$. If $\val(\Delta) \le \eps^{1000}$, i.e., is very small, then we know by the fact that $\val^{1/4}$ is a norm that $\val(\lambda gh) \approx \val(f)$. Then $\val(g)\val(h) \ge (1-o(1))\val(f)$, and hence one of $\val(g)$ or $\val(h)$ is much larger, as desired. Otherwise, $\val(\Delta)$ is still quite large and hence the $J'$-SVD of $\Delta$ for a random $J' \subseteq [n]$ has a large singular value. Let the first term in the $J'$-SVD of $\Delta$ be $\lambda'g'h'$, and write
\[ \Delta' = \Delta - \lambda'g'h' = f - \lambda gh - \lambda'g'h'. \] Again, note that $\|\Delta'\|_2^2 \le \|\Delta\|_2^2 - \eps^{O(1)}$. However, now $gh$ and $g'h'$ live on different partitions which makes it difficult to relate them to some SVD of $f$ (we want to use SVDs to increase the $\val$ like we did in \cref{subsec:overview99}). To handle this, we random restrict $f$ on a common refinement of $J$ and $J'$, i.e., $\bar{I} := (J \cap J') \cup (\bar{J} \cap \bar{J'})$. After random restriction,
\[ \Delta'_{I \to z} = f_{I \to z} - \lambda g_{I \to z}h_{I \to z} - \lambda'g'_{I \to z}h'_{I \to z}. \] Critically, $g_{I \to z}, h_{I \to z}, g'_{I \to z}, h'_{I \to z}$ all live on the same partition of $\bar{I}$. On average, the random restriction does not increase $\|\Delta'\|_2$ or decrease $\val(f)$ (see \cref{lemma:valrr}). Thus we may continue repeating this process until it terminates -- this should happen within $\eps^{-O(1)}$ steps because $\|\Delta'\|_2^2$ decreases during each step of the iteration. In the end, after several rounds of random restriction, we are left with a function $\Delta := f - \sum_{i=1}^T \lambda_i g_i h_i$, where $T \le \eps^{-O(1)}$, $g_i, h_i$ are functions over some partition of $f$, and $\val(\Delta)$ is very small. At this point, we are able to use a more complicated version of the argument in \cref{subsec:overview99} to build the desired function $g$ as a linear combination of the $g_i$'s.

There is one point to discuss: the $g_i, h_i$ may live on the same partition of $f$, but may not necessarily be the true SVD of $f$ on that partition. This is not difficult to fix: after each random restriction, replace the $g_i,h_i$ with the true SVD up to that many terms. This can only decrease $\|\Delta\|_2^2$ further, because \[ \min_{\lambda_i,g_i,h_i: 1 \le i \le t} \Big\|f - \sum_{i=1}^t \lambda_i g_i h_i \Big\|_2^2 \] is attained exactly when $\lambda_i, g_i, h_i$ are the first $t$ terms of the SVD of $f$.

\subsection{Getting Bounded Functions}
The above argument naturally gives rise to product functions $P$ that are only guaranteed to be $\ell_2$-bounded, whereas in~\Cref{thm:main} we require the product function to have absolute value $1$ always. In \cref{sec:boundedprob} we prove that if a $1$-bounded function $f$ correlates to an $\ell_2$-bounded product function, then $f$ in fact correlates to a $1$-bounded product function, with polynomial loss in the correlation; see~\Cref{thm:boundedprod} for a formal statement.

\subsection{The Global Inverse Theorem}
Finally, to get~\Cref{thm:main_global} we use the
following restriction inverse theorem, extending the ones from~\cite{BKM3,BKM4}:
\begin{restatable}{theorem}{ritre}\label{thm:rit}
Let $f: \A^n \to \bbC$ be $1$-bounded and $\mu = (1-\rho)\nu + \rho U$ where $U$ is uniform over $\A$. Suppose that for $\eps > 0$ it holds that
\[  \Pr_{I \sim_{1-\rho} [n], z \sim \nu^I}\left[\exists \{P_i: \A \to \bbC, \|P_i\|_\infty \le 1\}_{i \in \bar{I}} \enspace \text{ with } \enspace \Big|\E_{x \sim \A^{\bar{I}}}\Big[f_{I\to z}(x) \prod_{i \in \bar{I}} P_i(x_i) \Big]\Big| \ge \eps \right] \ge \eps. \]
Then there exist $1$-bounded $P_i: \A \to \bbC$ for $i = 1, \dots, n$, and a function $L: \A^n \to \bbC$ of degree at most $D := D(\rho,\eps)$ with $\|L\|_2 \le 1$ such that:
\[ \left|\E_{x \sim \mu^{\otimes n}}\left[f(x)L(x) \prod_{i=1}^n P_i(x_i) \right] \right| \ge \eps', \]
where quantitatively $\eps' \ge \exp(-C_{\A}(\rho^{-1}\log(1/\eps))^{O(1)})$ and $D \le C_{\A}(\eps\rho)^{-O(1)}$ for some constant $C_{\A}$.
\end{restatable}
In words,~\Cref{thm:rit} asserts that if a random restriction of $f$ is correlated with a product function with noticeable probability, then $f$ itself must be correlated with a product of a low-degree function with a product function.

The proof of~\Cref{thm:rit} proceeds along similar lines as in~\cite{BKM3,BKM4}, with some differences. In previous works, a key component of the proof was a small-set expansion statement, asserting that  if $h\colon \Sigma^n\to \Sigma^N$ is a function such that 
\[
\Pr_{\substack{x\sim \Sigma^n\\ y\sim_{\rho} x}}[h(x) = h(y)]\geq \eps
\]
(where $y\sim_{\rho} x $ is the distribution where for each coordinate $i$ independently, $y_i = x_i$ with probability $\rho$, and else $y_i\sim \Sigma$), then there exists $a\in \Sigma^N$ such that $h(x) = a$ for at least $\eps^{O_{\rho}(1)}$ of the $x$'s. Here, we have to prove a similar looking statement in the case that $h\colon \Sigma^n\to \mathbb{D}^N$ (where $\mathbb{D}$ is the unit disk), and we know that $h(x)$ and $h(y)$ are close in $\ell_2$ (as opposed to exactly equal) with probability at least $\eps$. In this case, we show that there exists $a\in (\mathbb{R}/\mathbb{Z})^N$ such that $h(x)$ is close in $\ell_2$ to $a$ for at least $\eps^{O_{\rho}(1)}$ of the $x$'s.

\section{The Swap Norm and Its Properties}
\label{sec:properties}
\subsection{The Box Norm}
We begin by discussing the \emph{box norm}, which is intimately related to the swap norm. 
\begin{definition}[Box inner product and box norm]
\label{def:box}
For functions $f_1,f_2,f_3,f_4: \A^n \to \bbC$ and $S \subseteq [n]$, 
define the box inner product as
\[ \boxx_S(f_1,f_2,f_3,f_4) := 
\E_{\substack{x_S,y_S \sim \A^S \\ x_{\bar{S}},y_{\bar{S}} \sim \A^{\bar{S}}}} 
\left[f_1(x_S,x_{\bar{S}})f_2(y_S,y_{\bar{S}})\bar{f_3(y_S,x_{\bar{S}})f_4(x_S,y_{\bar{S}})}\right]. \]
In the case that $f_1=f_2=f_3=f_4=f$, we simplify the notation and write $\boxx_S(f) = \boxx_S(f,f,f,f)$.
\end{definition}
It is easily seen that 
$\boxx_S(f_1,f_2,f_3,f_4)$ is a multi-linear form, 
and it is known that $\boxx_S(f)^{1/4}$ is a norm.
Indeed, the following lemma implies that:
\begin{lemma}
\label{lemma:boxstuff}
Let $f_1, f_2, f_3, f_4: \A^n \to \bbC$ and $S \subseteq [n]$. Then
\[ 
|\boxx_S(f_1,f_2,f_3,f_4)|^2 \le 
\boxx_S(f_1,f_3,f_1,f_3) \boxx_S(f_2,f_4,f_2,f_4), \]
\[ 
|\boxx_S(f_1,f_2,f_3,f_4)|^4 \le \boxx_S(f_1) \boxx_S(f_2)\boxx_S(f_3)\boxx_S(f_4), \]
and $\boxx_S(f) \le \|f\|_2^4$ for any $f: \A^n \to \bbC$.
\end{lemma}
\begin{proof}
For the first assertion,
\begin{align*}
    |\boxx_S(f_1,f_2,f_3,f_4)|^2 &= \left|\E_{x_S,y_S \sim \A^S} \Big(\E_{x \sim \A^{\bar{S}}} f_1(x_S, x) \bar{f_3(y_S, x)} \Big) \Big(\E_{x \sim \A^{\bar{S}}} \bar{f_4(x_S, x)} f_2(y_S, x) \Big) \right|^2 \\
    &\le \E_{x_S,y_S \sim \A^S} \Big|\E_{x \sim \A^{\bar{S}}} f_1(x_S, x) \bar{f_3(y_S, x)} \Big|^2 \E_{x_S,y_S \sim \A^S} \Big|\E_{x \sim \A^{\bar{S}}} \bar{f_4(x_S, x)} f_2(y_S, x) \Big|^2 \\
    &= \boxx_S(f_1, f_3, f_1, f_3) \boxx_S(f_2, f_4, f_2, f_4),
\end{align*}
where we have applied the Cauchy-Schwarz inequality.
Repeating this argument 
(but with the coupling $f_1,f_4$ and $f_2,f_3$) gives that 
$\boxx_S(f_1,f_2,f_3,f_4)^2\leq \boxx_S(f_1, f_4, f_1, f_4)\boxx_S(f_2,f_3,f_2,f_3)$, which gives 
$\boxx_S(f_1, f_3, f_1, f_3)^2\leq \boxx_S(f_1)\boxx_S(f_3)$, hence the second assertion.

For the third assertion we note that by Cauchy-Schwarz again
\begin{align*}
    \boxx_S(f,f,f,f) &= 
    \E_{\substack{x_S,y_S \sim \A^S \\ x_{\bar{S}},y_{\bar{S}} \sim \A^{\bar{S}}}} \left[f(x_S,x_{\bar{S}})f(y_S,y_{\bar{S}})\bar{f(y_S,x_{\bar{S}})f(x_S,y_{\bar{S}})}\right] \\
    &\le \left(\E_{\substack{x_S,y_S \sim \A^S \\ x_{\bar{S}},y_{\bar{S}} \sim \A^{\bar{S}}}} |f(x_S,x_{\bar{S}})|^2 |f(y_S,y_{\bar{S}})|^2 \E_{\substack{x_S,y_S \sim \A^S \\ x_{\bar{S}},y_{\bar{S}} \sim \A^{\bar{S}}}} |f(y_S,x_{\bar{S}})|^2 |f(x_S,y_{\bar{S}})|^2\right)^{1/2}\\ 
    &= \|f\|_2^4.
    \qedhere
\end{align*}
\end{proof}

\subsection{The Relation between the Box Norm and the Swap Norm}
Next, we observe that the swap inner product of 4 functions is the expected box inner product of the same 4 functions over a randomly chosen partition. 
\begin{lemma}
\label{lemma:box}
We have that
\[ \val(f_1,f_2,f_3,f_4) = \E_{S \subseteq [n]} \left[\boxx_S(f_1,f_2,f_3,f_4)\right]. \]
\end{lemma}
\begin{proof}
Looking at the definition of $\val(f_1,f_2,f_3,f_4)$, we denote by 
$S$ the random set of coordinates where $(x_i', y_i') = (y_i, x_i)$, i.e. where there was swap.
Conditioning on $S$, it can be seen that the expectation is equal to exactly $\boxx_S(f_1,f_2,f_3,f_4)$, and as the distribution of $S$ is uniform over all subsets of $[n]$ the claim follows.
\end{proof}

Using this connection,~\Cref{lemma:boxstuff} implies analogous properties for the swap norm.
\begin{lemma}
\label{lemma:l2}
We have that $\val(f_1,f_2,f_3,f_4) \le \|f_1\|_2\|f_2\|_2\|f_3\|_2\|f_4\|_2$.
\end{lemma}
\begin{proof}
By~\Cref{lemma:box}
\begin{align*}
\val(f_1,f_2,f_3,f_4) 
&= \E_{S \subseteq [n]} \left[\boxx_S(f_1,f_2,f_3,f_4)\right]\\
&\leq 
\E_{S \subseteq [n]} \left[\boxx_S(f_1)^{1/4}
\boxx_S(f_2)^{1/4}\boxx_S(f_3)^{1/4}
\boxx_S(f_4)^{1/4}\right]\\
&\leq 
\norm{f_1}_2
\norm{f_2}_2
\norm{f_3}_2
\norm{f_4}_2.\qedhere
\end{align*}
\end{proof}

We also get the following Cauchy-Schwarz-Gowers type inequalities for the swap norms.
\begin{lemma}
\label{lemma:cs1}
For any $f_1, f_2, f_3, f_4$ we have that:
\begin{equation} |\val(f_1,f_2,f_3,f_4)|^2 \le \val(f_1,f_2,f_1,f_2)\val(f_3,f_4,f_3,f_4) \label{eq:part1}
\end{equation} and
\begin{equation}
    \val(f_1,f_2,f_1,f_2)^2 \le \val(f_1,f_1,f_1,f_1)\val(f_2,f_2,f_2,f_2). \label{eq:part2}
\end{equation}
\end{lemma}
\begin{proof}
The first follows from the identity
\begin{align}
\val(f_1,f_2,f_3,f_4) 
&= \E_{x,y\sim\A^n}\left[f_1(x)f_2(y) \cdot \E_{(x',y') \sim (x \lr y)} \bar{f_3(x')f_4(y')}\right]\notag\\
&=  \E_{x,y\sim\A^n} 
\left[\E_{(x',y') \sim (x \lr y)} f_1(x')f_2(y') \cdot \E_{(x',y') \sim (x \lr y)} \bar{f_3(x')f_4(y')}\right].
\label{eq:iden}
\end{align}
The second equality follows because $\E_{(x',y') \sim (x \lr y)} \bar{f_3(x')f_4(y')}$ is the same for all $(x'', y'')$ in the support of $(x \lr y)$. Now, \eqref{eq:part1} follows from the Cauchy-Schwarz inequality. For~\eqref{eq:part2}, we note:
\begin{align*}
    \val(f_1,f_2,f_1,f_2)^2 &= \left( \E_S \boxx_S(f_1,f_2,f_1,f_2) \right)^2\\ 
    &\le \left(\E_S \boxx_S(f_1,f_1,f_1,f_1)^{1/2} \boxx_S(f_2,f_2,f_2,f_2)^{1/2} \right)^2 \\
    &\le \E_S \boxx_S(f_1,f_1,f_1,f_1) \E_S \boxx_S(f_2,f_2,f_2,f_2) \\
    &= \val(f_1,f_1,f_1,f_1) \val(f_2,f_2,f_2,f_2).
\end{align*}
Here, the first equality is \cref{lemma:box}, and the first inequality is \cref{lemma:boxstuff}.
\end{proof}
Applying \cref{lemma:cs1} implies that $\val(f)^{1/4}$ is a norm.
\begin{corollary}
\label{cor:norm}
$\val(f)^{1/4}$ is a norm on functions $f$.
\end{corollary}
\begin{proof}
Clearly it obeys scaling and nonnegativity by \eqref{eq:iden}. By \cref{lemma:cs1} we get that:
\begin{align*}
    \val(f_1+f_2) &= \sum_{(a,b,c,d)\in\{1,2\}^4} \val(f_a,f_b,f_c,f_d) \\
    &\le \sum_{(a,b,c,d)\in\{1,2\}^4} \val(f_a)^{1/4}\val(f_b)^{1/4}\val(f_c)^{1/4}\val(f_d)^{1/4}
    \\ &= (\val(f_1)^{1/4} + \val(f_2)^{1/4})^4. \enspace \qedhere
\end{align*}
\end{proof}

The next lemma proves that $\val(f)^{1/2}$ is non-decreasing in expectation under random restrictions.
\begin{lemma}
\label{lemma:valrr}
For a subset $I \subseteq [n]$ and $f: \A^n \to \bbC$ we have that
\[ \val(f,f,f,f)^{1/2} \le \E_{z \sim \A^I}\left[\val(f_{I\to z},f_{I\to z},f_{I\to z},f_{I\to z})^{1/2} \right]. \]
\end{lemma}
\begin{proof}
Note that
\begin{align*} \val(f,f,f,f) &= \E_{\substack{x\sim\A^I,y\sim\A^I \\ (x',y') \sim (x \lr y)}} \val(f_{I\to x}, f_{I\to y}, f_{I \to x'}, f_{I \to y'}) \\
&\le \E_{\substack{x\sim\A^I,y\sim\A^I \\ (x',y') \sim (x \lr y)}} \val(f_{I\to x})^{1/4} \val(f_{I\to y})^{1/4} \val(f_{I \to x'})^{1/4} \val(f_{I \to y'})^{1/4} \\
&\le \left(\E_{\substack{x\sim\A^I,y\sim\A^I \\ (x',y') \sim (x \lr y)}} \val(f_{I\to x})^{1/2} \val(f_{I\to y})^{1/2} \E_{\substack{x\sim\A^I,y\sim\A^I \\ (x',y') \sim (x \lr y)}} \val(f_{I\to x'})^{1/2} \val(f_{I\to y'})^{1/2}\right)^{1/2} \\
&= \E_{\substack{x\sim\A^I,y\sim\A^I \\ (x',y') \sim (x \lr y)}} \val(f_{I\to x})^{1/2} \val(f_{I\to y})^{1/2} = \left(\E_{x\sim\A^I} \val(f_{I\to x})^{1/2}\right)^2.
\end{align*}
Here the first inequality uses \cref{lemma:cs1} and the second uses Cauchy-Schwarz.
\end{proof}

\subsection{3-Wise Correlation Implies Large Swap Norm}
\label{sec:3wise}

The goal of this section is to establish \cref{thm:toswap}.

\subsubsection{Path Tricks}
\label{sec:pathtrick}
In this section we discuss the \emph{path trick}, which is a tool that allows us to modify the underlying distribution $\mu$ to get additional desirable properties from it.
We start by defining the distribution resulting from a path trick of $r$ steps on $\mu$.
\begin{definition}
\label{def:pathtrick}
Given a distribution $\mu$ on $\A \times \B \times \C$ and integer $r \ge 1$, define the \emph{$x$ path trick} distribution $\mu^+$ on $\A^+ \times \B \times \C$ (where $\A^+ \subseteq \A^{2^r-1}$) as follows:
\begin{enumerate}
    \item Sample $y_1 \sim \mu_y$.
    \item Sample $(x_1, z_1)$ from $\mu$, conditioned on $y = y_1$.
    \item Sample $(x_1', y_2)$ from $\mu$, conditioned on $z = z_1$.
    \item Repeat steps 2 and 3, and in the final step sample $(x_{2^r}, y_{2^{r-1}+1})$ from $\mu$ conditioned on $z = z_{2^{r-1}}$.
\end{enumerate}
$\mu^+$ is over the sequences $((x_1, x_1', \dots, x_{2^{r-1}-1}', x_{2^{r-1}}), y_1, z_{2^{r-1}})$ generated by the process above.
\end{definition}
The $y$ and $z$ path tricks can be defined symmetrically. One way to think of the $x$-path trick is as considering the bipartite graph between $\Gamma$ and $\Phi$ that has edges between $(y,z)$ if there is $x$ such that $(x,y,z)$ is in the support of $\mu$, and in that case the edge is labeled by $x$. Path tricks correspond to walks in this bipartite graph, and the resulting symbol in the first coordinate corresponds to the edge labels collected throughout the walk.
The following result from \cite{BKM1} asserts that if $f,g,h$ have non-trivial correlation with respect to $\mu$, then $f^{+}$, $g$ and $h'$ have non-trivial correlation with respect to the path trick distribution $\mu^{+}$, where $f^{+}$ is an alternating conjugated product of $f$ along the labels and $h'$ may be an arbitrary $1$-bounded function. We reproduce its proof below for the reader's convenience.
\begin{lemma}[Path trick]
\label{lemma:pathtrick}
Let $\mu$ be a distribution on $\A \times \B \times \C$. If $1$-bounded functions $f: \A^n \to \bbC$, $g: \B^n \to \bbC$ and $h: \C^n \to \bbC$ satisfy $|\E_{(x,y,z) \sim \mu^{\otimes n}}[f(x)g(y)h(z)] \ge \eps$, then for any $r \ge 1$ and $\mu^+$ defined as above for an $r$-step path trick, there are $1$-bounded functions $f^+: (\A^+)^n \to \bbC$, $g: \B^n \to \bbC$, $\wt{h}: \C^n \to \bbC$ satisfying:
\[ \Big|\E_{(x,y,z) \sim (\mu^+)^{\otimes n}}[f^+(x)g(y)\wt{h}(z)]\Big| \ge \eps^{2^r}. \]
Additionally for $x = (x_1, x_1', \dots, x_{2^{r-1}}) \in (\A^+)^n$ it holds that $f^+(x) = \prod_{i=1}^{2^{r-1}} f(x_i) \prod_{i=1}^{2^{r-1}-1} \bar{f(x_i')}$.
\end{lemma}
\begin{proof}
We start by performing a simple reduction on $h$. Define \[ \wt{h}(z) := \E_{(x,y,z') \sim \mu^{\otimes n}}[\bar{f(x)g(y)} | z' = z]. \] Then note that
\begin{align*} \left| \E_{(x,y,z)\sim\mu^{\otimes n}}[f(x)g(y)h(z)] \right| &= \left| \E_{z\sim\mu_z}[h(z)\bar{\wt{h}(z)}] \right| \le \|h\|_2 \|\wt{h}\|_2 \le \|\wt{h}\|_2 = \E_{(x,y,z)\sim\mu^{\otimes n}}[f(x)g(y)\wt{h}(z)]^{1/2}. \end{align*}
Hence $\left| \E_{(x,y,z)\sim\mu^{\otimes n}}[f(x)g(y)\wt{h}(z)] \right| \ge \eps^2$. The remainder of the proof involves repeated application of Cauchy-Schwarz. We show the following by induction on $r$:
\begin{align} \left| \mathop{\E}_{\substack{x_1,x_1',\dots,x_{2^{r-1}}, x_{2^{r-1}}' \\ y_1, \dots, y_{2^{r-1}+1} \\ z_1, \dots, z_{2^{r-1}}}}\left[\prod_{i=1}^{2^{r-1}} f(x_i) \bar{f(x_i')} \cdot g(y_1) \bar{g(y_{2^{r-1} + 1})} \right] \right| \ge \eps^{2^r}. \label{eq:biginduct} \end{align}
Let us check the base case $r = 1$.
\begin{align*}
\eps^2 &\le \left| \E_{(x,y,z)\sim\mu^{\otimes n}}[f(x)g(y)\wt{h}(z)] \right| = \left| \mathop{\E}_{\substack{z_1 \\ x_1, y_1, x_1', y_2}}[f(x_1) \bar{f(x_1')} g(y_1) \bar{g(y_2)}] \right|.
\end{align*}
The inductive step will follow by squaring this, and using Cauchy-Schwarz on the variables other than $y_{2^{r-1}+1}$. Precisely,
\begin{align*}
    \eps^{2^{r+1}} &\le \left| \mathop{\E}_{\substack{x_1,x_1',\dots,x_{2^{r-1}}, x_{2^{r-1}}' \\ y_1, \dots, y_{2^{r-1}+1} \\ z_1, \dots, z_{2^{r-1}}}}\left[\prod_{i=1}^{2^{r-1}} f(x_i) \bar{f(x_i')} \cdot g(y_1) \bar{g(y_{2^{r-1}+1})} \right] \right|^2 \\
    &\le \mathop{\E}_{y_{2^{r-1}+1}} \left| \mathop{\E}_{\substack{x_1,x_1',\dots,x_{2^{r-1}}, x_{2^{r-1}}' \\ y_1, \dots, y_{2^{r-1}} \\ z_1, \dots, z_{2^{r-1}}}}\left[\prod_{i=1}^{2^{r-1}} f(x_i) \bar{f(x_i')} \cdot g(y_1) \right] \right|^2.
\end{align*}
Note that this last expression exactly equals \eqref{eq:biginduct} for $r+1$, completing the induction. Finally, the definition of $\wt{h}$ gives that the LHS expression in \eqref{eq:biginduct} can be expressed as:
\begin{align*} \left| \mathop{\E}_{\substack{x_1,x_1',\dots,x_{2^{r-1}} \\ y_1, \dots, y_{2^{r-1}} \\ z_1, \dots, z_{2^{r-1}}}}\left[\prod_{i=1}^{2^{r-1}} f(x_i) \prod_{i=1}^{2^{r-1}-1} \bar{f(x_i')} \cdot g(y_1) \wt{h}(z_{2^{r-1}}) \right] \right|.
\end{align*}
So we can set $f^+(x) = \prod_{i=1}^{2^{r-1}} f(x_i) \prod_{i=1}^{2^{r-1}-1} \bar{f(x_i')}$. This is exactly what we wanted to prove.
\end{proof}

\paragraph{Pairwise-connectedness:} We remark that it is easy to observe that if 
$\mu$ is pairwise-connected, then $\mu^{+}$ is also pairwise-connected. In fact, it improves connectivity between two of the coordinates: if $r$ is taken so that $2^r$ is larger than $|\Gamma|,|\Phi|$, then $\supp(\mu^{+}_{y,z})=\Gamma\times \Phi$; see~\cite[Lemma 3.12]{BKM4} for a proof. This fact will be useful for us later on.

\subsubsection{A Symmetry Argument}
\label{subsec:symmetry}

Now we prove \cref{thm:toswap} by applying a sequence of path tricks, and then applying a symmetry argument. We start by applying the following sequence of three path tricks.
\begin{enumerate}
    \item Apply a $y$ path trick with parameter $r$ satisfying $2^r \ge \max\{|\A|, |\C|\}$.
    \item Apply a $z$ path trick with parameter $r$ satisfying $2^r \ge \max\{|\A|, |\B^+|\}$.
    \item Apply a $x$ path trick with parameter $r = 2$.
\end{enumerate}
After this sequence of operations, we know that there is a distribution $\mu^+$ on $\A^+ \times \B^+ \times \C^+$ such that $\A^+ \subseteq \A^3$, and $1$-bounded functions $g: (\B^+)^n \to \bbC$, $h: (\C^+)^n \to \bbC$ such that
\begin{equation}
    \left|\E_{\substack{(x,y,z) \sim (\mu^+)^n \\ x = (x_1,x_2,x_3)}} f(x_1)\bar{f(x_2)}f(x_3)g(y)h(z) \right| \ge \eps^{O(1)}. \label{eq:3path}
\end{equation}
By doing the path tricks in this specific order, we are able to guarantee the following nice property of $\supp(\mu^+)$ that allows us to then exploit symmetry.
\begin{lemma}
\label{lemma:suppmu}
Let distribution $\mu^+$ over $\A^+ \times \B^+ \times \C^+$ be as defined above. For any $a, b \in \A$ there are $y \in \B^+, z \in \C^+$ such that both $((a,a,b),y,z) \in \supp(\mu^+)$ and $((b,a,a),y,z) \in \supp(\mu^+)$.
\end{lemma}
\begin{proof}
After the first $y$ path trick, by pairwise connectivity the distribution $\mu'$ over $\A \times \B^+ \times \C$ is pairwise-connected and satisfies that $\mu'_{xz}$ has full support. Let $z \in \C$ be chosen arbitrarily, and let $y_a, y_b$ be such that $(a, y_a, z) \in \supp(\mu')$ and $(b, y_b, z) \in \supp(\mu')$. Let $\mu''$ be the distribution over $\A \times \B^+ \times \C^+$ after applying the $z$ path trick to $\mu'$. Again by pairwise connectivity, $\mu''$ is pairwise-connected and the support of $\mu''_{xy}$ is full, so there is some $z' \in \C^+$ such that $(a, y_b, z') \in \supp(\mu'')$. Also, note that for $\vec{z} := (z, z, \dots, z)$ we have that that $(a, y_a, \vec{z}), (b, y_b, \vec{z}) \in \supp(\mu'')$. To conclude, set $y := y_b$ and $z := \vec{z}$ and note that $(a, a, b)$ corresponds to the path $y_b \to z' \to y_b \to \vec{z}$, and $(b, a, a)$ corresponds to the path $y_b \to \vec{z} \to y_a \to \vec{z}$.
\end{proof}

We are now ready to prove \cref{thm:toswap}, restated below.
\toswap*
\begin{proof}
Consider $\mu^{+}$ as above, and define the distribution $\mu^-$ as follows. Sample $a, b \in \A$ uniformly at random and sample $c \in \{a, b\}$ uniformly at random. Now let $y, z$ be as guaranteed by \cref{lemma:suppmu}, i.e., $((a,c,b),y,z) \in \supp(\mu^+)$ and $((b,c,a),y,z) \in \supp(\mu^+)$. Include $((a,c,b), y, z)$ in $\mu^-$ with probability mass $\frac{1}{2|\A|^2}$. Let $\alpha > 0$ and $\nu$ be such that $\mu^+ = \alpha \mu^- + (1-\alpha)\nu$. Applying a random restriction and using \cref{eq:3path} tells us that
\begin{equation}\label{eq:symmetry_argument}
\E_{\substack{I \sim_{1-\alpha} [n] \\ ((u_1,u_2,u_3),v,w) \sim \nu^I}} \left|\E_{\substack{(x,y,z) \sim (\mu^-)^{\bar{I}} \\ x = (a, c, b)}} f_{I \to u_1}(a)\bar{f_{I \to u_2}(c)}f_{I \to u_3}(b)g_{I \to v}(y)h_{I \to w}(z) \right| \ge \eps^{O(1)}. 
\end{equation}
Define the function $\wt{f}_{I,u}: \supp(\mu^-_x)^{\bar{I}} \to \bbC$ as:
\[ \wt{f}_{I,u}(a,c,b) := \E_{a',b' \sim (a \lr b)} f_{I \to u_1}(a')\bar{f_{I \to u_2}(c)}f_{I \to u_3}(b'). 
\]
Note that by choice of $y, z$ (via \cref{lemma:suppmu}) we have that that
\[ \E_{\substack{(x,y,z) \sim (\mu^-)^{\bar{I}} \\ x = (a, c, b)}} f_{I \to u_1}(a)\bar{f_{I \to u_2}(c)}f_{I \to u_3}(b)g_{I \to v}(y)h_{I \to w}(z) = \E_{\substack{(x,y,z) \sim (\mu^-)^{\bar{I}} \\ x = (a, c, b)}} \tilde{f}_{I,u}(a,c,b)g_{I \to v}(y)h_{I \to w}(z). \]
Because $g, h$ are $1$-bounded this, together with~\eqref{eq:symmetry_argument}, implies that
\begin{align*}
    \eps^{O(1)} &\le \E_{\substack{I \sim_{1-\alpha} [n] \\ (u_1,u_2,u_3) \sim \nu_x^I}} \left\|\tilde{f}_{I,u} \right\|_{L^2((\mu^-_x)^{\bar{I}})} \\
    &= \E_{\substack{I \sim_{1-\alpha} [n] \\ (u_1,u_2,u_3) \sim \nu_x^I}} \left(
    \E_{(a,c,b) \sim (\mu^-_x)^{\bar{I}}} 
    \left[|f_{I\to u_2}(c)|^2 \left|\E_{(a',b') \sim (a \lr b)} f_{I\to u_1}(a')f_{I \to u_3}(b')\right|^2\right] \right)^{1/2} \\
    &\leq \E_{\substack{I \sim_{1-\alpha} [n] \\ (u_1,u_2,u_3) \sim \nu_x^I}} \left(
    \E_{(a,c,b) \sim (\mu^-_x)^{\bar{I}}} 
    \left[\left|\E_{(a',b') \sim (a \lr b)} f_{I\to u_1}(a')f_{I \to u_3}(b')\right|^2\right] \right)^{1/2} \\
    &= \E_{\substack{I \sim_{1-\alpha} [n] \\ (u_1,u_2,u_3) \sim \nu_x^I}} \val(f_{I\to u_1}, f_{I \to u_3}, f_{I \to u_1}, f_{I \to u_3})^{1/2} \\
    &\le \left(\E_{\substack{I \sim_{1-\alpha} [n] \\ (u_1,u_2,u_3) \sim \nu_x^I}} \val(f_{I \to u_1}) \right)^{1/4} \left(\E_{\substack{I \sim_{1-\alpha} [n] \\ (u_1,u_2,u_3) \sim \nu_x^I}} \val(f_{I \to u_3}) \right)^{1/4},
\end{align*}
where we have applied Cauchy-Schwarz and \cref{lemma:cs1} \eqref{eq:part2}. To conclude, note that for $(x_1,x_2,x_3) \sim \mu^+_x$ that both $x_1$ and $x_3$ are distributed as $\mu_x$, and thus if $\nu'$ is the distribution of $u_1$ (or $u_3$) for $(u_1,u_2,u_3) \sim \nu$, then $(1-\alpha)\nu' + \alpha U = \mu_x$. This exactly implies the conclusion of \cref{thm:toswap}.
\end{proof}

\section{Proof of~\Cref{thm:swapnine}: the 99\%~Regime}
\label{sec:99}
The goal of the section is to establish~\Cref{thm:swapnine}. The key component in the proof is the following lemma, which then implies \cref{thm:swapnine} by induction.
\begin{lemma}
\label{lemma:99induct}
Let $f\colon \Sigma^n\to\mathbb{C}$ be a function such that $\norm{f}_2\leq 1$ and $\val(f) = 1-\eps$ for $\eps\leq 0.05$. Let $\lambda_1$ be the maximum singular value under the $\{1\}$-SVD of $f$. Then:
\begin{itemize}
    \item $\lambda_1^2 \ge 1-3\eps$.
    \item If $\delta = 1-\lambda_1^2$ and $f = \sum_{i=1}^m \lambda_ig_ih_i$ is the $\{1\}$-SVD then $\val(h_1) \ge (1-\eps)(1+\Omega(\delta))$.
\end{itemize}
\end{lemma}

We require a weak bound on $\val$ applied to orthogonal functions.
\begin{lemma}
\label{lemma:valperp}
If $\l h_1, h_2\r = 0$ and $\|h_1\|_2 = \|h_2\|_2 = 1$, then $\val(h_1,h_2,h_1,h_2) \le 2/3$.
\end{lemma}
\begin{proof}
Let $\alpha,\beta$ be uniform complex numbers with $|\alpha| = |\beta| = 1/\sqrt{2}$. Note that
\begin{align*}
    1 &\ge \E_{\alpha,\beta}[\val(\alpha h_1 + \beta h_2, \alpha h_1 + \beta h_2, \alpha h_1 + \beta h_2, \alpha h_1 + \beta h_2)] \\ &= \frac14 \val(h_1) + \frac14 \val(h_2) + \val(h_1,h_2,h_1,h_2) \\
    &\ge \frac12 \sqrt{\val(h_1)\val(h_2)} + \val(h_1,h_2,h_1,h_2) \\
    &\ge \frac32 \val(h_1,h_2,h_1,h_2),
\end{align*}
where the first inequality uses \cref{lemma:l2} and the last uses \cref{lemma:cs1}.
\end{proof}
Now we can proceed to the proof of \cref{lemma:99induct}.
\begin{proof}[Proof of \cref{lemma:99induct}]
Writing $f = \sum_i \lambda_i g_ih_i$ in the $\{1\}$-SVD, one can check that
\[
1-\eps = \val(f,f,f,f) = \sum_i \lambda_i^4 \val(h_i) + 2\sum_{i < j} \lambda_i^2\lambda_j^2 \val(h_i,h_j,h_i,h_j). 
\]
Applying \cref{lemma:valperp} gives
\begin{equation}\label{eq:base_case1}
    1-\eps \le \sum_i \lambda_i^4\val(h_i) + 2 \sum_{i<j} \lambda_i^2\lambda_j^2 \cdot \frac{2}{3} = \sum_i \lambda_i^4\val(h_i) + \frac{2}{3}\left(\left(\sum_i \lambda_i^2\right)^2 - \sum_i \lambda_i^4\right),
\end{equation}
which is at most $\frac{1}{3} \sum_i \lambda_i^4 + \frac{2}{3}$ as $\val(h_i)\leq 1$.
Thus $\sum_i \lambda_i^4 \ge 1-3\eps$, and as $\sum_{i}\lambda_i^2 \leq 1$ we conclude that $\lambda_1^2 \ge 1-3\eps$. Denoting $\delta = 1 - \lambda_1^2$, we get that $\delta \le 3\eps$ and also that $\sum_{j \neq 1} \lambda_j^2 \leq \delta$.
Using these notations we conclude from~\eqref{eq:base_case1} that
\[
1-\eps \leq \sum_i \lambda_i^4 \val(h_i) + \frac{2}{3}\left(1- \sum_i \lambda_i^4\right) \le \lambda_1^4 \val(h_1) + \frac43\delta + O(\delta^2),
\]
where the final inequality uses $\lambda_1^4 \ge 1-2\delta$. Thus,
\[ \val(h_1) \ge \frac{1-\eps-\frac43\delta-O(\delta^2)}{(1-\delta)^2} \ge (1-\eps)(1+\Omega(\delta)), \]
for sufficiently small $\eps$ (recall that $\delta \le 3\eps$). This completes the proof.
\end{proof}
We now prove \cref{thm:swapnine}, restated below.
\swapnine*
\begin{proof}
We prove the following claim by induction on $n$: there is a constant $C$ such that for all sufficiently small $\eps$, if $\val(f) \ge 1-\eps$ then there is a product function $P = P_1 \dots P_n$ with $\|P\|_2 = 1$ and $\l f, P \r \ge 1-C\eps$. It is evident for $n = 1$.

Now let $n > 1$. Perform the SVD as in \cref{lemma:99induct} and let $\lambda_1^2 = 1-\delta$. We will set $P = g_1P'$ where $\l h_1, P' \r \ge 1 - C(\eps-\Omega(\delta))$ by induction, because $\val(h_1) \ge (1-\eps)(1+\Omega(\delta))$. Note that
\[ \l f, P \r = \lambda_1 \l h_1, P' \r \ge (1-\delta)^{1/2}(1 - C(\eps-\Omega(\delta))) \ge 1-C\eps \] 
where the last inequality holds provided that $C$ is a sufficiently large constant.
\end{proof}

\section{Proof of~\Cref{thm:swap}: the Inverse Theorem for the Swap Norm}
The goal of this section is to establish \cref{thm:swap}. At a high level, the proof of the
result proceeds by induction on $\eps$. We show that
given a function $f$, we can do a process consisting of taking random restrictions and SVD decompositions to get a function $f'$ whose (normalized) value is significantly larger than that of $f$. Additionally, the function $f'$ will have the property that if it is correlated with a product function, then $f$ is also correlated with a product function (perhaps a different one). Here and throughout, the normalized value of $f$ is $\val(f)^{1/2}/\norm{f}_2^2$; it will be easier to work with as we are considering restrictions of $f$ that may have different $2$-norms.

\subsection{An Iterative Scheme}
\label{subsec:iter}
Fix a function $f: \A^n \to \bbC$ with $\sqrt{\val(f)} \ge \eps^2$ as in the statement of~\cref{thm:swap}. In this section we provide an iterative argument that alternately peels off large singular vectors for random partitions and performs random restrictions, with the goal of increasing $\val(f)$ so that we can apply induction. At each time step in the iteration we maintain a subset $R_t \subseteq [n]$ (a random subset of some density), a partition $R_t = S_t \cup T_t$ (which is uniormly random conditioned on $R_t$), and a function $f_t: \A^{R_t} \to \bbC$ which is a random restriction of $f$ onto $R_t$.

Let us now formally define the iterative scheme. 
Set $f_0 = f$, $t=0$, $R_0 = [n]$, and take $S_0, T_0$ as a uniformly random partition of $R_0$. We choose $\gamma = \Theta(\eps^8)$, to be determined more precisely later.
\begin{enumerate}
\item Define $\Delta_t := f_t - \sum_{i=1}^t \lambda_ig_ih_i$, where $\lambda_1, \dots, \lambda_t$ are the $t$ largest singular values of the $S_t$-SVD of $f_t$ with corresponding singular vectors $g_i: \A^{S_t} \to \bbC, h_i: \A^{T_t} \to \bbC$.
\item If $\val(\Delta_t)^{1/2} < \gamma\|f_t\|_2^2$ or $t \ge 8\gamma^{-2}$, terminate.
\item If $\val(\Delta_t)^{1/2} \ge \gamma \|f_t\|_2^2$ and $t < 8\gamma^{-2}$, let $R_t = S' \cup T'$ be a uniformly random partition of $R_t$. Set $S_{t+1} := S_t \cap S'$ and $T_{t+1} := T_t \cap T'$. Set $R_{t+1} := S_{t+1} \cup T_{t+1}$.
\item Let $f_{t+1} := (f_t)_{(R_t\setminus R_{t+1}) \to z}$ for random $z \sim \A^{R_t\setminus R_{t+1}}$. Increase $t$ by $1$ and iterate.
\end{enumerate}

We use the following potential function to analyze the progress of the iterative scheme:
\[ \Phi_t := \frac{\eps^{-5}\val(f_t)^{1/2} - \|\Delta_t\|_2^2}{\|f_t\|_2^2}. \]
The idea behind the potential function is that during each iteration, if it does not terminate, $\Delta_t$ is subtracting off one extra term in the SVD, so $\|\Delta_t\|_2^2$ should be decreasing and thus increases $\Phi_t$. Thus, if $\val(f_t)$ stays the same throughout the process, we will eventually have that $\val(\Delta_t)$ must be small. 

The first step in the analysis is to argue that with non-negligible probability, $\Phi_t$ does not decrease significantly under a random restriction.
\begin{lemma}
\label{lemma:phirr}
Let $f, g : \A^n \to \bbC$ be functions such that $f$ is $B$-bounded. Let $I \subseteq [n]$, $C > 1$, and $\delta > 0$. Then with probability at least $\frac{\delta\|f\|_2^2}{CB^2}$ over $z \sim \A^I$:
\begin{equation} \frac{C\val(f_{I\to z})^{1/2} - \|g_{I\to z}\|_2^2}{\|f_{I\to z}\|_2^2} \ge \frac{C\val(f)^{1/2} - \|g\|_2^2}{\|f\|_2^2} - 2\delta \enspace \text{ and } \enspace \|f_{I\to z}\|_2^2 \ge \frac{\delta}{C+2\delta+\frac{\|g\|_2^2}{\|f\|_2^2}}\|f\|_2^2. \label{eq:phirr} \end{equation}
\end{lemma}
\begin{proof}
Let $\Phi := \frac{C\val(f)^{1/2} - \|g\|_2^2}{\|f\|_2^2}$. By \cref{lemma:valrr} we know that:
\[
\E_{z \sim \A^I}[C\val(f_{I\to z})^{1/2} - \|g_{I\to z}\|_2^2 - (\Phi - 2\delta)\|f_{I\to z}\|_2^2]\geq
C\val(f)^{1/2}-\|g\|_2^2 - (\Phi - 2\delta)\|f\|_2^2 
= 2\delta\|f\|_2^2. 
\]
As the random variable under the expectation on the left hand side is at most $CB^2$ always (since $f$ is $B$-bounded), it follows from an averaging argument that
\[ \Pr_{z \sim \A^I}\left[C\val(f_{I\to z})^{1/2} - \|g_{I\to z}\|_2^2 - (\Phi - 2\delta)\|f_{I\to z}\|_2^2 \ge \delta\|f\|_2^2\right] \ge \frac{\delta\|f\|_2^2}{CB^2}. \]
Any $z$ for which this event holds satisfies the conclusion of the lemma. Indeed, the first inequality in~\eqref{eq:phirr} follows by bounding $\delta\norm{f}_2^2\geq 0$ and re-arranging, and the
second inequality follows by bounding $\val(f_{I\to z})^{1/2} \le \|f_{I\to z}\|_2^2$ (which holds by~\cref{lemma:l2}), $\norm{g_{I\rightarrow z}}_2^2\geq 0$, $\Phi \ge -\|g\|_2^2/\|f\|_2^2$
and rearranging.
\end{proof}

\subsubsection{Analyzing a Single Iteration}
We now state and prove the main iteration lemma, 
stating that each step is successful with noticeable probability, wherein a successful step is one in which the potential $\Phi_t$ increases significantly and the $2$-norm of $f_t$ does not decrease too much.
\begin{lemma}
\label{lemma:iter}
Let $f_t, R_t = S_t \cup T_t, \Delta_t$ be defined as above. If $\val(\Delta_t)^{1/2} \ge \gamma\|f_t\|_2^2$, then with probability at least $\Omega(\gamma^4\eps^5\|f_t\|_2^6/B^6)$, it holds that $\Phi_{t+1} \ge \Phi_t + \gamma^2/4$ and $\|f_{t+1}\|_2^2 \ge \Omega(\eps^5\gamma^2)\|f_t\|_2^2$.
\end{lemma}
\begin{proof}
Let $R_t = S' \cup T'$ be a uniformly random partition, as described in the iterative scheme. Define $\tilde{\Delta} := \Delta_t/\|\Delta_t\|_2$, so that $\|\tilde{\Delta}\|_2 = 1$ and $\val(\tilde{\Delta}) \ge \gamma^2\|f_t\|_2^4\|\Delta_t\|_2^{-4}$ by the assumption that the process has not terminated. By \cref{lemma:box} we have that $\E_S[\boxx_S(\wt{\Delta})] = \val(\wt{\Delta})
\geq \gamma^2\|f_t\|_2^4\|\Delta_t\|_2^{-4}$, and as $\boxx_S(\wt{\Delta}) \le 1$ for all $S$ by \cref{lemma:boxstuff}, we get by an averaging argument that
\[ 
\Pr_{S'}
\left[\boxx_{S'}(\tilde{\Delta}) 
\ge \frac{\gamma^2\|f_t\|_2^4}{2\|\Delta_t\|_2^{4}}\right] \ge \frac{\gamma^2\|f_t\|_2^4}{2\|\Delta_t\|_2^{4}} 
\ge \frac{\gamma^2\|f_t\|_2^4}{2B^{4}}. 
\] 
In the last inequality we used the fact that $\|\Delta_t\|_2^4 \le \|f_t\|_2^4 \le B^4$. Let $\kappa_1 \ge \kappa_2 \ge \dots$ be the singular values of $\tilde{\Delta} \in \bbC^{\A^{S'} \times \A^{T'}}$. With these notations, we get that with probability at least $\gamma^2\|f_t\|_2^4B^{-4}/2$ over $S'$
\begin{equation}\label{eq:iterative_1}
\kappa_1^2 \ge \sum_i \kappa_i^4 = \boxx_{S'}(\tilde{\Delta}) \ge \frac{\gamma^2\|f_t\|_2^4}{2\|\Delta_t\|_2^{4}}. 
\end{equation}
Let $g: \A^{S'} \to \bbC$ and $h: \A^{T'} \to \bbC$ be the unit singular vectors corresponding to $\kappa_1$ and define
$\Delta_t' := \Delta_t - \|\Delta_t\|_2\kappa_1gh$. Using~\eqref{eq:iterative_1} it follows that
\begin{equation} \|\Delta_t'\|_2^2 = \|\Delta_t\|_2^2 - \|\Delta_t\|_2^2\kappa_1^2 \le \|\Delta_t\|_2^2 - \frac{\gamma^2}{2}\frac{\|f_t\|_2^4}{\|\Delta_t\|_2^{2}} \le \|\Delta_t\|_2^2 - \frac{\gamma^2}{2}\|f_t\|_2^2, \label{eq:deltadrop} \end{equation} 
where the last transition uses the fact that $\|\Delta_t\|_2 \le \|f_t\|_2$.

Now we will take the random restriction and invoke \cref{lemma:phirr}. Let $I := R_t \setminus R_{t+1}$, $\delta = \gamma^2/8$, $C = \eps^{-5}$, and use \cref{lemma:phirr} for $f := f_t, g := \Delta_t'$ to get that with probability at $\frac{\gamma^2\eps^5\|f_t\|_2^2}{8B^2}$ over $z \sim \A^I$,
\[ \frac{\eps^{-5}\val((f_t)_{I\to z})^{1/2} - \|(\Delta_t')_{I\to z}\|_2^2}{\|(f_t)_{I\to z}\|_2^2} \ge \frac{\eps^{-5}\val(f_t)^{1/2} - \|\Delta_t'\|_2^2}{\|f_t\|_2^2} - \frac{\gamma^2}{4}, \] and
\[ \|f_{t+1}\|_2^2 = \|(f_t)_{I \to z}\|_2^2 \ge \Omega(\gamma^2\eps^5)\|f_t\|_2^2 
\] 
(we used $\|\Delta_t'\|_2^2 \le \|f_t\|_2^2$ to replace the $\norm{g}_2^2/\norm{f}_2$ term in~\cref{lemma:phirr} by $1$). 
By \eqref{eq:deltadrop} we know that $\frac{\eps^{-5}\val(f_t)^{1/2} - \|\Delta_t'\|_2^2}{\|f_t\|_2^2} \ge \Phi_t + \frac{\gamma^2}{4}$. Thus, the conclusion follows from the fact that $f_{t+1} = (f_t)_{I\to z}$ and $\|(\Delta_t')_{I\to z}\|_2^2 \ge \|\Delta_{t+1}\|_2^2$. Indeed, the latter inequality follows because
\[ (\Delta_t')_{I\to z} = f_{t+1} - \sum_{i=1}^t \lambda_i(g_i)_{I\to z}(h_i)_{I\to z} - \|\Delta_t\|_2\kappa_1g_{I\to z}h_{I\to z}, \]
and $(g_1)_{I \to z}, \dots, (g_t)_{I \to z}$ and $g_{I \to z}$ are all functions from $\A^{S_{t+1}} \to \bbC$, and $(h_1)_{I \to z}, \dots, (h_t)_{I \to z}$ and $h_{I \to z}$ are all functions from $\A^{T_{t+1}} \to \bbC$, and that the SVD is the optimal way to reduce the $\ell_2$ norm of a matrix by subtracting rank one matrices.
\end{proof}

\subsubsection{Analyzing the Termination State}
We next analyze the functions $f_t$ and $\Delta_t$ upon the termination of the process. In the following lemma, we show that with noticeable probability, either the process terminated with 
$\Delta_t$ with small value while $\val(f_t)^{1/2}/\norm{f_t}^2$ has not decrease by much, or else  $\val(f_t)^{1/2}/\norm{f_t}^2$ is noticeably larger than $\val(f)^{1/2}/\norm{f}^2$.
Denote by $\tau$ the time step on which the iteration scheme ends, so either $\tau = 8\gamma^{-2}$ or $\val(\Delta_{\tau})^{1/2} < \gamma\|f_{\tau}\|_2^2$.
\begin{lemma}
\label{lemma:tlarge}
Let $T = 8\gamma^{-2}$. Over the randomized process described above, with probability at least $(\eps\|f\|_2/B)^{O(T^2)}$, we either have (\textbf{Case 1}) that:
\[ \frac{\val(f_T)^{1/2}}{\|f_T\|_2^2} \ge \frac{\val(f)^{1/2}}{\|f\|_2^2} + \eps^5 \enspace \text{ and } \enspace \|f_T\|_2^2 \ge (\eps\gamma)^{O(T)}\|f\|_2^2, \]
or (\textbf{Case 2}) there is some $\tau \in \{1, 2, \dots, T\}$ such that $\val(\Delta_\tau)^{1/2} < \gamma\|f_\tau\|_2^2$, $\|f_t\|_2^2 \ge (\eps\gamma)^{O(t)}$ for all $t \le \tau$, and
\[ \frac{\val(f_\tau)^{1/2}}{\|f_\tau\|_2^2} \ge \frac{\val(f)^{1/2}}{\|f\|_2^2} - \eps^5. \]
\end{lemma}
\begin{proof}
Consider the event that the conclusion of \cref{lemma:iter} is true at each iteration $i \le \tau$. Under this we get that $\|f_i\|_2^2 \ge (\eps\gamma)^{O(i)}\|f\|_2^2$ for all $i = 0, 1, \dots, T$. The probability of this is at least $(\eps\gamma/B)^{O(T)} \cdot \prod_{i=0}^t \|f_i\|_2^{O(1)} \ge (\eps\gamma\|f\|_2/B)^{O(T^2)}$.

In this event, we get that $\Phi_t \ge \Phi_0 + t\gamma^2/4$, for $t \le \tau$.
Thus
\[ \eps^{-5}\left(\frac{\val(f_t)^{1/2}}{\|f_t\|_2^2} - \frac{\val(f)^{1/2}}{\|f\|_2^2} \right) \ge t\frac{\gamma^2}{4} + \frac{\|\Delta_t\|_2^2}{\|f_t\|_2^2} - \frac{\|\Delta_0\|_2^2}{\|f\|_2^2} \ge t\frac{\gamma^2}{4} - 1, \] because $\Delta_0 = f$. If $t = 8\gamma^{-2}$, we are in Case 1. Otherwise, Case 2 holds.
\end{proof}
Intuitively, if the first case of Lemma~\ref{lemma:tlarge} holds then 
we have successfully made progress by increasing the value of $f_t$. At that point we could appeal to the inductive hypothesis and conclude that $f_t$ is correlated with a product function, and thus be done. Otherwise, the second case holds, and
we know that we have not lost too much in value when passing from $f$ to $f_t$ but at the same time we know that $\val(\Delta_t)$ is small, meaning that, in a sense, we have exhausted all of the meaningful parts in the SVD decompositions of $f$. In the next subsection, we show how to use this information to construct a function $f'$ related to $f$ that has significantly larger value, such that correlations of $f'$ with product functions are related to correlation of $f$ with product functions.
\subsubsection{Analyzing {\bf Case 2} in~\Cref{lemma:tlarge}}
In this section we provide a way to increment $\val(f)$ in the second case in~\cref{lemma:tlarge}. 

\begin{lemma}
\label{lemma:svdinc}
If $\val(\Delta_t)^{1/2} \le c^2\eps^8\|f_t\|_2^2$, for sufficiently small $c$, and $\frac{\val(f_t)^{1/2}}{\|f_t\|_2^2} \in [\eps^2/2, 0.99]$ for some $t = O(\eps^{-16})$, then there are $B\eps^{-O(1)}\|f_t\|_2^{-1}$-bounded functions $g: \A^{S_t} \to \bbC, h: \A^{T_t} \to \bbC$ with $\|g\|_2 = \|h\|_2 = 1$ and constant $\eta_{g,h}$ with $|\eta_{g,h}| \ge \Omega(\eps^8)$ such that at least one of the following holds:
\begin{enumerate}
\item $\val(g) \ge (1+\Omega(\eps^2))\frac{\val(f_t)}{\|f_t\|_2^4}$ and for any  $I \subseteq S_t$, $z \in \A^I$, and $P: \A^{S_t \setminus I} \to \bbC$ it holds that \[ \l (f_t)_{I \to z}, Ph\r = \eta_{g,h}\|f_t\|_2\l g_{I\to z}, P \r. \]
\item $\val(h) \ge (1+\Omega(\eps^2))\frac{\val(f_t)}{\|f_t\|_2^4}$ and for any $I \subseteq T_t$, $z \in \A^I$, and $P: \A^{T_t \setminus I} \to \bbC$ it holds that
\[ \l (f_t)_{I \to z}, gP \r = \eta_{g,h} \|f_t\|_2\l h_{I \to z}, P \r. \]
\end{enumerate}
\end{lemma}
Before proving this, we need a technical lemma, asserting that the singular vectors of large singular values are bounded in $\ell_\infty$. This helps to prove that $g, h$ in \cref{lemma:svdinc} are bounded.
\begin{lemma}
\label{lemma:svdinf}
Let $f \in \bbC^{S \times T}$ and let $g \in \bbC^S$ and $h \in \bbC^T$ with $\|g\|_2 = \|h\|_2 = 1$ be singular vectors of $f$ with singular value $\lambda$. Then $g, h$ are $B/\lambda$-bounded.
\end{lemma}
\begin{proof}
Thinking of $f$ as a matrix and $g$ and $h$ as vectors, we have that $g = \frac{1}{\lambda}f^*h$. Thus
\[ 
\|g\|_\infty \le \frac{B\|h\|_1}{\lambda} \le \frac{B\|h\|_2}{\lambda} = \frac{B}{\lambda}. 
\qedhere
\]
\end{proof}
\begin{proof}[Proof of \cref{lemma:svdinc}]
By scaling, we may assume that $\|f_t\|_2 = 1$. For simplicity of notation, let $f := f_t$, consider the $S_T$-SVD decomposition of $f$, say 
$f=\sum\limits_{i}\lambda_ig_i h_i$ where 
$\lambda_{i}\geq\lambda_{i+1}$ for all $i$. 
Define
\[ 
\Delta := f - \sum_{i \in [t], \lambda_i \ge \frac{c\eps^4}{t}} \lambda_i g_ih_i 
\] 
and recall that 
$\Delta_t = f-\sum\limits_{i\in [t]}\lambda_ig_ih_i$.
By \cref{cor:norm}, we know that
\[ 
\val(\Delta)^{1/4} \le 
\val(\Delta_t)^{1/4} + \val(\Delta-\Delta_t)^{1/4}
\leq 
c\eps^4+\sum\limits_{i\in [t],\lambda_i<\frac{c\eps^4}{t}} \lambda_i
\le 2c\eps^4. 
\] 
Thus, $\Delta$ has a small value and we henceforth work with it. For simplicity, relabel $t$ as the number of $\lambda_i$ subtracted in the definition of $\Delta$, so that $\Delta = f - \sum_{i=1}^t \lambda_ig_ih_i$.  We remark that (by the upper bound on $t$ in the statement) we get that $\lambda_i\geq \Omega(c\eps^{20})$ for each $i=1,\ldots,t$.

Let $F = f - \Delta = \sum_{i=1}^t \lambda_ig_ih_i$. By \cref{cor:norm} we know that 
\begin{equation}
    \val(F)^{1/4} \ge \val(f)^{1/4} - \val(\Delta)^{1/4} \ge \val(f)^{1/4}(1-4c\eps^2). \label{eq:lowerF}
\end{equation}
The remaining analysis is split into two cases depending on the size of $\lambda_1$.

\paragraph{Case 1: $\lambda_1 \ge 1 - \eps^2/10^{10}$.} By \cref{lemma:l2,cor:norm} we know that
\begin{align*}
\val(\lambda_1g_1h_1)^{1/4} &\ge \val(F)^{1/4} - \val\left(\sum_{i=2}^t \lambda_ig_ih_i \right)^{1/4}
\\&\ge \val(f)^{1/4}(1-4c\eps^2) - \left\|\sum_{i=2}^t \lambda_ig_ih_i \right\|_2 \\
&\ge 0.999\val(f)^{1/4},
\end{align*}
where we have used that $\left\|\sum_{i=2}^t \lambda_ig_ih_i \right\|_2 \le \sqrt{1-\lambda_1^2} \le \eps/10^5$.
Because $\val(g_1h_1) = \val(g_1)\val(h_1)$ we conclude that because $\val(f) \le 0.99$, \[ \val(g_1) \ge (1+\Omega(1))\val(f) \enspace \text{ or } \enspace \val(h_1) \ge (1+\Omega(1))\val(f). \] Assume the latter by symmetry. In this case, let $g = g_1, h = h_1$. Then item 2 follows because
\begin{align*}
    \l f_{I\to z}, gP \r &= \Big\l \sum_i \lambda_i g_i(h_i)_{I \to z}, g_1P \Big\r = \lambda_1 \l h_{I \to z}, P \r.
\end{align*}
and $g_1, h_1$ are $O(B)$-bounded by \cref{lemma:svdinf}, and $\eta_{g,h} := \lambda_1$.

\paragraph{Case 2: $\lambda_1 \le 1 - \eps^2/10^{10}$.} Let $\delta = 1-\lambda_1^2$. Define
\[ \val_{T_t}(F, F, F, F) := \E_{\substack{x,y \sim \A^{R_t} \\ (x',y') \sim (x\lr_{T_t}y)}} F(x)F(y)\bar{F(x')}\bar{F(y')}, \] where $(x',y') \sim (x \lr_{T_t} y)$ denotes that $(x_i',y_i') = (x_i,y_i)$ for $i \in R_t \setminus T_t$, and otherwise $(x_i',y_i') = (x_i,y_i)$ or $(y_i,x_i)$ with probability $1/2$ for $i \in T_t$. Then we can see that
\begin{align}
\val(f)(1-16c\eps^2) \le \val(F) \le \val_{T_t}(F) = \sum_{i=1}^t \lambda_i^4 \val(h_i,h_i,h_i,h_i) + 2\sum_{1\le i<j \le t} \lambda_i^2\lambda_j^2 \val(h_i,h_j,h_i,h_j). \label{eq:valf}
\end{align}
Here, the first inequality is by \eqref{eq:lowerF}, and the second follows because (using \eqref{eq:iden})
\begin{align*}
    \val(f,f,f,f) &= \E_{x,y\sim\A^n} \left|\E_{(x',y') \sim (x \lr y)} f_1(x')f_2(y') \right|^2 \\
    &= \E_{x,y\sim\A^n} \left|\E_{(x',y') \sim (x \lr y)} \E_{(x'',y'') \sim (x' \lr_{T_t} y')} f_1(x'')f_2(y'') \right|^2 \\
    &\le \E_{x,y\sim\A^n} \E_{(x',y') \sim (x \lr y)} \left|\E_{(x'',y'') \sim (x' \lr_{T_t} y')} f_1(x'')f_2(y'') \right|^2 \\
    &= \E_{x,y\sim\A^n} \left|\E_{(x',y') \sim (x \lr_{T_t} y)} f_1(x')f_2(y') \right|^2\\ 
    &= \val_{T_t}(f,f,f,f).
\end{align*}
We further split into two cases depending on the size of the second term in~\eqref{eq:valf}.
\paragraph{Case 2(a): The second term is large.} Consider the case that
\[ 2\sum_{1\le i<j \le t} \lambda_i^2\lambda_j^2 \val(h_i,h_j,h_i,h_j) \ge \frac{\delta}{4}\val(f). \]
In this case let $\beta_1, \dots, \beta_t \in \bbC$ be i.i.d.~unit complex numbers, let $\alpha_i = t^{-1/2}\beta_i$, and $g := \sum_{i=1}^t \bar{\alpha_i} g_i$. Then $\|g\|_2 = 1$ and
\begin{align*}
    \l f_{I \to z}, gP \r 
    = \Big\l \sum_i \lambda_i g_i(h_i)_{I \to z}, \Big(\sum_{i=1}^t \overline{\alpha_i} g_i\Big)P \Big\r 
    = \Big\l \sum_{i=1}^t \lambda_i\alpha_i(h_i)_{I \to z}, P \Big\r = t^{-1/2}\l h_{I \to z}, P \r
\end{align*} for $h := \sum_{i=1}^t \beta_i\lambda_ih_i$, so $\|h\|_2 \le 1$, so the second part of item 2 holds with $\eta_{g,h} = t^{-1/2} \ge \Omega(\eps^8)$. A calculation yields that
\begin{align*}
\E_{\beta_1,\dots,\beta_t} \val(h) &= \sum_{i,j,k,\ell \in [t]} \lambda_i\lambda_j\bar{\lambda_k\lambda_{\ell}} \val(h_i,h_j,h_k,h_{\ell}) \E[\beta_i\beta_j\bar{\beta_k \beta_{\ell}}] \\
&= \sum_{i=1}^t \lambda_i^4 \val(h_i,h_i,h_i,h_i) + 4\sum_{1\le i<j \le t} \lambda_i^2\lambda_j^2 \val(h_i,h_j,h_i,h_j) \\
&\ge \val(f)(1-16c\eps^2) + 2\sum_{1\le i<j \le t} \lambda_i^2\lambda_j^2 \val(h_i,h_j,h_i,h_j) \\
&\ge \val(f)(1-16c\eps^2) + \frac{\delta}{4}\val(f) \ge \val(f)(1+\Omega(\eps^2)),
\end{align*}
for sufficiently small $c$. Here, the second equality follows because $\E[\beta_i\beta_j\bar{\beta_k \beta_{\ell}}] = 1$ only if $\{i,j\} = \{k,\ell\}$, and otherwise is $0$. Also, the first inequality uses \eqref{eq:valf}. Thus, item 2 holds for some choice of $\beta_1, \dots, \beta_t$. The functions $g$ and $h$ are $B\eps^{-O(1)}$ bounded by \cref{lemma:svdinf} because $\lambda_i \ge \eps^{O(1)}$ for all $i$.

\paragraph{Case 2(b): The second term is small.} Precisely, that
\[ 2\sum_{1\le i<j \le t} \lambda_i^2\lambda_j^2 \val(h_i,h_j,h_i,h_j) \le \frac{\delta}{4}\val(f). \]
In this case, we know that
\[ \sum_{i=1}^t \lambda_i^4 \val(h_i,h_i,h_i,h_i) \ge \val(f)(1-16c\eps^2) - \frac{\delta}{4}\val(f) \ge (1-\delta/3)\val(f). \]
Now we remove all small $\lambda_i$ to get that
\[ \sum_{1\le i\le t, \lambda_i^2 > \frac{\delta\val(f)}{10}} \lambda_i^4 \val(h_i,h_i,h_i,h_i) \ge (1-\delta/3)\val(f) - \delta\val(f)/10 \ge (1-\delta/2)\val(f). \]
To conclude, note that
\[ \max_{i:\lambda_i^2 > \frac{\delta\val(f)}{10}} \val(h_i) \ge \frac{(1-\delta/2)\val(f)}{\sum_i \lambda_i^4} \ge \frac{(1-\delta/2)\val(f)}{\lambda_1^2} \ge (1+\Omega(\delta))\val(f). \]
Thus, picking $i$ achieving the maximum on the left hand side, we may set $g = g_i, h = h_i$ and recall from above that $\l f_{I \to z}, gP \r = \lambda_i \l (h_i)_{I \to z}, P \r$, where $\lambda_i^2 \ge \Omega(\eps^6)$. Thus the second part of item 2 holds for $\eta_{g,h} := \lambda_i \ge \Omega(\eps^3)$. Finally, $g_i, h_i$ are $B\eps^{-O(1)}$ bounded by \cref{lemma:svdinf}.
\end{proof}

\subsection{Main Induction: Proving \cref{thm:swap}}
\label{subsec:swapinduct}
Now we apply the iteration scheme of \cref{subsec:iter} to prove \cref{thm:swap}, restated below, by an inductive argument.
\swap*
\begin{proof}
We prove the following statement by induction on $\eps$ and $B$: under the hypotheses of \cref{thm:swap}, there is a constant $\delta(\eps, B)$ such that
\begin{align} \E_{\substack{I\sim_{1-\delta(\eps, B)} [n] \\ z \sim \A^I}} \left[ \sup\Big\{\Big|\Big\l f_{I\to z}, \prod_{i\in \bar{I}}P_i \Big\r \Big| : P_i: \A \to \bbC, \|P_i\|_2 \le 1 \text{ for } i \in \bar{I} \Big\} \right] \ge \delta(\eps, B). \label{eq:inductswap} \end{align}
\cref{thm:swap} follows from \eqref{eq:inductswap}, because $f$ is $B$-bounded. If $\eps \ge 0.99$ then \eqref{eq:inductswap} follows from \cref{thm:swapnine}. For the inductive step, note that for each $i = 0, 1, \dots, t$, the distribution of set $R_i$ is $R_i\sim_{2^{-i}} [n]$. Also, conditioned on $R_i$, $S_i \sim_{1/2} R_i$.
Let $\gamma = c\eps^8$ for a sufficiently small constant $c$. The conclusion of \cref{lemma:tlarge} holds with probability at least $q(\eps, B) := (\eps/B)^{O(T^2)}$. In Case 1 of \cref{lemma:tlarge}, let $\tau := 8\gamma^{-2}$, and in Case 2 let $\tau$ be as in the statement.

We first consider Case 1, when $\tau = 8\gamma^{-2}$. If so then \[ \frac{\val(f_\tau)^{1/2}}{\|f_\tau\|_2^2} \ge \frac{\val(f)^{1/2}}{\|f\|_2^2} + \eps^5 \ge \eps^2 + \eps^5 \] by \cref{lemma:tlarge}. Let $\eps' := (\eps^2 + \eps^5)^{1/2} = \eps + \Omega(\eps^4)$ and $B' := B\|f\|_2/\|f_\tau\|_2 \ge B\eps^{-O(\tau)}$ by \cref{lemma:tlarge}. By induction, a random restriction of $f_\tau$ correlates to a product function, i.e.,
\begin{align*} 
\E_{\substack{I\sim_{1-\delta(\eps', B')} R_t \\ z \sim \A^I}} \left[ \sup\Big\{\Big|\Big\l (f_\tau)_{I\to z}, \prod_{i\in R_t \setminus I}P_i \Big\r \Big| : P_i: \A \to \bbC, \|P_i\|_2 \le 1 \text{ for } i \in R_t \setminus I \Big\} \right] &\ge \delta(\eps', B')\|f_t\|_2 \\
&\ge \delta(\eps', B')\eps^{O(\tau)}.
\end{align*}
Because restricting onto $(R_t \setminus I) \sim_{\delta(\eps',B')} R_\tau$ is equivalent in distribution to restricting onto a set sampled $\sim_{2^{-\tau}\delta(\eps',B')} [n]$, we conclude that:
\[ \E_{\substack{I\sim_{1-2^{-\tau}\delta(\eps', B')} [n] \\ z \sim \A^I}} \left[ \sup\Big\{\Big|\Big\l f_{I\to z}, \prod_{i\in \bar{I}}P_i \Big\r \Big| : P_i: \A \to \bbC, \|P_i\|_2 \le 1 \text{ for } i \in \bar{I} \Big\} \right] \ge q(\eps, B)\delta(\eps', B')\eps^{O(\tau)}. \] Thus this case follows by induction for the choice (where $\tau \le \eps^{-O(1)}$) for
\[ \delta(\eps, B) := q(\eps, B)\eps^{O(\tau)}\delta(\eps+\Omega(\eps^4), B\eps^{-O(\tau)}). \]
In Case 2, by symmetry we may assume that item 2 of \cref{lemma:svdinc} holds. Then with probability at least $q(\eps, B)/2$ over $R_t \sim_{2^{-\tau}} [n], z \sim \A^{[n]\setminus R_t}$, and $S_t \sim_{1/2} R_t$, that there are functions $g, h$ satisfying item 2 of \cref{lemma:svdinc}.
By item 2 of \cref{lemma:svdinc} we know that
\[ \val(h) \ge (1+\Omega(\eps^2)) \frac{\val(f_\tau)}{\|f_\tau\|_2^4} \ge (1+\Omega(\eps^2))\left(\frac{\val(f)^{1/2}}{\|f\|_2^2} - \eps^5\right)^2 \ge (1+\Omega(\eps^2))\val(f), \] as $\val(f)^{1/2} \ge \eps^2$ and $\|f\|_2 = 1$. Let $\eps' = \eps + \Omega(\eps^3)$ and $B' = B\eps^{-O(\tau)}$ so that $h$ is $B'$-bounded (recall that $\|f_\tau\|_2^2 \ge \eps^{O(\tau)}$). By induction we know that a random restriction of $h$ correlates to a product function, i.e.,
\begin{align} \E_{\substack{I\sim_{1-\delta(\eps', B')} T_t \\ z \sim \A^I}} \left[ \sup\Big\{\Big|\Big\l h_{I\to z}, \prod_{i\in T_t \setminus I}P_i \Big\r \Big| : P_i: \A \to \bbC, \|P_i\|_2 \le 1 \text{ for } i \in T_t \setminus I \Big\} \right] \ge \delta(\eps', B'). \label{eq:ttrr} \end{align}
In particular, note that the random restriction in the above equation is independent of $S_t$. For a fixed restriction $I\sim_{1-\delta(\eps', B')} T_t, z \sim \A^I$ let $P_{I, z} := \prod_{i \in T_t \setminus I} P_i$ denote the product function guaranteed in \eqref{eq:ttrr}. Then, we know that:
\begin{align*} 
\E_{\substack{I\sim_{1-\delta(\eps', B')} T_t \\ z \sim \A^I}} \E_{z' \sim \A^{S_t}} |\l (f_\tau)_{(S_t, I) \to (z', z)}, g_{S_t \to z'}P_{I,z} \r| 
&\ge \E_{\substack{I\sim_{1-\delta(\eps', B')} T_t \\ z \sim \A^I}} |\l (f_\tau)_{I \to z}, gP_{I,z} \r| \\
\\ 
&= \eta_{g,h}\|f_\tau\|_2 \E_{\substack{I\sim_{1-\delta(\eps', B')} T_t \\ z \sim \A^I}} |\l h_{I\to z}, P_{I,z} \r| \\
&\ge \eps^{O(\tau+1)}\delta(\eps', B'),
\end{align*}
where the first transition is by the triangle inequality, the second transition is by \cref{lemma:svdinc}, and the third one is by the choice of $P_{I,z}$. The result follows because:
\begin{enumerate}
    \item The hypotheses required to apply \cref{lemma:svdinc} occur with probability at least $q(\eps, B)$.
    \item $g_{S_t \to z'}$ is a constant bounded by $B\eps^{-O(\tau)}$ (by \cref{lemma:svdinc}), and
    \item $(f_\tau)_{(S_t, I) \to (z', z)}$ has identical distribution as $f_{J \to x}$ for $J \sim_{1-2^{-(\tau+1)}\delta(\eps',B')}$ and $x \sim \A^J$.
\end{enumerate}
This completes the induction, by setting $\delta(\eps, B) = B^{-1}\eps^{O(\tau+1)}q(\eps, B)\delta(\eps', B')$. Solving the recursion for $\delta$ gives the quantitative bounds claimed in the theorem.
\end{proof}

\section{Correlation to Bounded Product Functions}
\label{sec:boundedprob}
Note that \cref{thm:swap} says that (under random restriction) there is an $\ell_2$-bounded product function that correlates with the $1$-bounded function $f$. In this section, we prove that this random restriction must in fact correlate to a $1$-bounded product function.
\begin{theorem}
\label{thm:boundedprod}
Let $f: \A^n \to \bbC$ be a function with 
$\|f\|_4\leq 1$ and $P_i: \A \to \bbC$ be such that $\|P_i\|_2 \le 1$ and
\[ \Big|\E_{x\sim\A^n}[f(x) \prod_{i=1}^n P_i(x_i)] \Big| \ge \delta. \]
Then there are functions $P_i': \A \to \D$ such that
\[ \Big|\E_{x\sim\A^n}[f(x) \prod_{i=1}^n P_i'(x_i)] \Big| \ge \delta^{O(1)}. \]
\end{theorem}
Fix $f$ and $P$ as in the setting of the theorem.
By making small perturbations in $P$ we may assume that $P(x) \neq 0$ for all $x$.
Our proof relies on the following standard approximation of intervals by exponential sums.
\begin{lemma}
\label{lemma:exp}
For any real numbers $R \ge L$ and $\eta$ there is $c: \R \to \bbC$ and function $b(x) := \int_{-\infty}^\infty c(\theta)e^{2\pi i\theta x}$ satisfies:
\begin{itemize}
    \item $\int_{-\infty}^\infty |c(\theta)| d\theta \le \left(\frac{R-L}{\eta}\right)^{1/2}$, and
    \item For $x \in [L+\eta/2, R-\eta/2]$ it holds that $b(x) = 1$
    \item For $x \in (-\infty, L-\eta/2] \cup [R+\eta/2, \infty)$ it holds that $b(x) = 0$,
    \item For all $x \in \R$ it holds that $0 \le b(x) \le 1$.
\end{itemize}
\end{lemma}
\begin{proof}
Let $I_{[A,B]}$ be the indicator function of the interval $[A,B]$. Let $b = \eta^{-1} I_{[L,R]} * I_{[-\eta/2,\eta/2]}$. All items except the first are evident. For the first item, for a function $g: \R \to \bbC$ let $\hat{g}(\theta) := \int_{-\infty}^{\infty} g(x) e^{-2\pi i\theta x} dx$ be the Fourier transform. By Fourier inversion,
\begin{align*}
    \int_{-\infty}^\infty |c(\theta)| d\theta &= \frac{1}{\eta}\int_{-\infty}^\infty |\hat{I_{[L,R]}}(\theta)|
    |\hat{I_{[-\eta/2,\eta/2]}}(\theta)| d\theta \\
    &\le \frac{1}{\eta}\left(\int_{-\infty}^\infty 
    |\hat{I_{[L,R]}}(\theta)|^2 d\theta \right)^{1/2} \left(\int_{-\infty}^\infty |\hat{I_{[-\eta/2,\eta/2]}}(\theta)|^2 d\theta \right)^{1/2} = \left(\frac{R-L}{\eta}\right)^{1/2},
\end{align*}
by Parseval's identity.
\end{proof}
Now we move towards the proof of \cref{thm:boundedprod}. Write $P_i = Q_i R_i$ where $R_i: \A \to \R_{>0}$, and $|Q_i(x)| = 1$ for all $x \in \A$. By replacing $f \to f\prod_{i=1}^n Q_i$ we may assume that $P_i = R_i$ in fact, i.e., $P_i$ take nonnegative real values.

Define the set $S_{L,R} := \{x \in \A^n : e^L \le P(x) \le e^R\}$ and let $\mu(S_{L,R}) := \Pr_{x \in \A^n}[x \in S_{L,R}]$ and $A_{L,R} := \E_{x \in \A^n}[f(x)P(x)1_{x\in S_{L,R}}]$.
Let $D = \log(16/\delta^2)$.
\begin{lemma}
\label{lemma:ad}
It holds that $A_{-D,D} \ge \delta/2$.
\end{lemma}
\begin{proof}
    First we bound
    \[ \left|\E_{x \sim \A^n}[f(x)P(x)1_{P(x) \le \delta^2/16}] \right| \le \frac{\delta}{4}, \] because $\|f\|_1\leq 1$. 
    Also, by H\"{o}lder's inequality and $\norm{f}_4\leq 1$
    \[ 
    \left|\E_{x \sim \A^n}[f(x)P(x)1_{P(x) \le 16/\delta^2}] \right| 
    \leq
    \left|\E_{x \sim \A^n}[|P(x)|^{4/3}1_{P(x) \le 16/\delta^2}]\right|^{3/4} 
    \le \frac{\delta}{4}\E_{x \sim \A^n} |P(x)|^2 \le \frac{\delta}{4}, \] because $f$ is $1$-bounded and $\|P\|_2 \le 1$.
\end{proof}

\begin{proof}[Proof of \cref{thm:boundedprod}]
Let $\alpha, \eta < 1$ be parameters to be chosen later. Let $-D-1 \le L \le D$ and $R = L+\alpha$. Let $c_{L,R}(\theta)$ be the coefficients guaranteed by \cref{lemma:exp}. Then
\begin{align*}
    \int_{-\infty}^\infty c_{L,R}(\theta) \left(\E_{x\sim\A^n}[f(x) e^{2\pi i\theta \ln P(x)}] \right) d\theta = \E_{x \sim \A^n} [f(x)g_{L,R}(x)],
\end{align*}
where $g_{L,R}(x) = b(\ln P(x))$ and $b$ is the function in \cref{lemma:exp}. By the properties of $g_{L,R}$,
\begin{align*}
    \left|\E_{x \in \A^n} [f(x)g_{L,R}(x)]\right| \ge \left|\E_{\substack{x \sim \A^n}} f(x)1_{x\in S_{L,R}} \right| - \mu(S_{L-\eta/2,L+\eta/2}) - \mu(S_{R-\eta/2,R+\eta/2}).
\end{align*}
Furthermore, we can bound (using $-D-1 \le L \le D$)
\begin{align*}
    \left|\E_{\substack{x \sim \A^n }} f(x)1_{x\in S_{L,R}} \right| &\ge e^{-L}\left|\E_{\substack{x \sim \A^n }} f(x)P(x)1_{x\in S_{L,R}} \right| - \E_{\substack{x \sim \A^n }} \left[|e^{-L}P(x)-1| \cdot |f(x)|1_{x\in S_{L,R}} \right] \\
    &\ge e^{-L}A_{L,R} - 2\alpha \mu(S_{L,R}),
\end{align*}
because $|e^{\alpha}-1| \le 2\alpha$ (for $\alpha \le 1$) and $\norm{f}_1\leq 1$. Overall, we have concluded that
\begin{align*}
&\int_{-\infty}^\infty c_{L,R}(\theta)\left(\E_{x\sim\A^n}[f(x) e^{2\pi i\theta \ln P(x)}] \right) d\theta\\
&\qquad\qquad\qquad\ge e^{-L}A_{L,R} - 2\alpha \mu(S_{L,R}) - \mu(S_{L-\eta/2,L+\eta/2}) - \mu(S_{R-\eta/2,R+\eta/2}) . 
\end{align*}
The expectation of the RHS over $L \sim [-D-1,D]$ chosen uniformly at random is at least
\begin{equation}\label{eq:1bdd_corr1}
\frac{\alpha A_{-D,D}}{e^{-D}(2D+1)} - \frac{4\alpha^2}{2D+1} - \frac{4\eta}{2D+1} \ge \frac{\alpha\delta^3}{32(2D+1)} - \frac{4\alpha^2}{2D+1} - \frac{4\eta}{2D+1}, 
\end{equation} 
where the last inequality is by~\cref{lemma:ad}. Here,
we used the fact that 
\[
\E_{L}\left[e^{-L}A_{L,R}\right]
\geq e^{-D}\E_{L}[A_{L,R}]
\geq e^{-D} A_{-D,D} \frac{\alpha}{2D+1},
\]
and that similar calculations give
\[
\E_{L}\left[\mu(S_{L,R})\right]
\leq \frac{2\alpha}{2D+1},
\qquad
\E_{L}\left[\mu(S_{L-\eta/2,L+\eta/2})\right]
\leq \frac{2\eta}{2D+1},
\qquad
\E_{L}\left[\mu(S_{R-\eta/2,R+\eta/2})\right]
\leq 
\frac{2\eta}{2D+1}.
\]

Choosing $\alpha = \delta^3/1000$ and $\eta = \delta^6/400000$ the expression in~\eqref{eq:1bdd_corr1} is at least $\frac{\delta^6}{100000(2D+1)}$. Hence
\[ \int_{-\infty}^\infty c_{L,R}(\theta) \left(\E_{x\sim\A^n}[f(x) e^{2\pi i\theta \ln P(x)}] \right) d\theta \ge \frac{\delta^6}{100000(2D+1)} \] for some choice of $L$. Because $\int_{-\infty}^\infty |c_{L,R}(\theta)| d\theta \le (\alpha/\eta)^{1/2} = 20/\delta$, we conclude that
\[ \left|\E_{x\sim\A^n}[f(x) e^{2\pi i\theta \ln P(x)}] \right| \ge \frac{\delta^7}{2000000(2D+1)} \] for some $\theta$. This concludes the proof, because $e^{2\pi i\theta \ln P(x)}$ is a product function in which each component is into the unit disc.
\end{proof}

\paragraph{Proof of~\cref{thm:main}:}
\cref{thm:main} immediately follows by combining \cref{thm:toswap,thm:swap,thm:boundedprod}.

\section{Restriction Inverse Theorem}
\label{sec:rit}

The goal of this section is to establish~\Cref{thm:rit}.

\subsection{Preliminaries on Product Functions}
\label{subsec:setup}
In this section we collect preliminary facts about product functions and their \emph{distance}. For $I \subseteq [n]$ define $\mc{F}(I)$ to be the set of product functions $P: \A^I \to \bbC$ such that $|P_i(x)| = 1$ for all $x \in \A^n$ and $i\in [n]$.
We start by discussing how to convert such a product function $P: \A^n \to \bbC$ into a vector.

Throughout this section we will use the notation $O_{\rho,\gamma}(\cdot)$ to denote that we are suppressing constants depending on $\rho$, $\gamma$, with polynomial dependence on the parameters being suppressed. For example, if $\rho$, $\gamma < 1$, this means $O((\rho\gamma)^{-O(1)})$. It will be convenient for us to fix a special reference input $\a\in \Sigma$.
\begin{definition}
\label{def:vector}
Let $P(x) = \prod_{i=1}^n P_i(x_i)$ be a product function with $P_i(\a) \in \R_{\ge0}$ for all $i \in [n]$. Define $\pi(P) \in \R^{2n|A|}$ as $\pi(P)_{i,\a'} = P_i(\a') \in \bbC$. For general $\wt{P} = \prod_{i=1}^n \wt{P}_i$, let $c_1, \dots, c_n \in \bbC$ with $|c_i| = 1$ be such that $c_i\wt{P}_i(\a) \in \R_{\ge0}$. Then define $\pi(\wt{P}) := \pi(\prod_{i=1}^n c_i \wt{P})$.
\end{definition}
Informally, $\pi(P)$ is defined by concatenating $P_i(\a')$ for $i \in [n], \a' \in \A$, if $P_i(\a)$ are real. Note that $\pi(P)$ is a $2n|A|$ dimensional because $P_i(\a')$ are complex numbers. The reason we are taking the reference point $\a$ and normalizing appropriately so that the $P(\a)$ is real valued is to ensure that $\pi(P)$ is invariant under rotation by a unit complex number. Our next goal is to establish that $|\l P, Q\r|$ for product functions $P, Q$ is related to $\|\pi(P) - \pi(Q)\|_2$. The statement bellow asserts that if $\pi(P)$ and $\Pi(Q)$ are far, then $P$ and $Q$ can only have a small correlation. Throughout this section, the implicit constant $\Omega$ hides factors that depend on the probability of the smallest atom in $\mu$.


\begin{lemma}
\label{lemma:prodip}
Let $P, Q: \A^n \to \bbC$ be $1$-bounded product functions. Then \[ |\l P, Q\r| \le \exp(-\Omega(\|\pi(P)-\pi(Q)\|_2^2)). \]
\end{lemma}
\begin{proof}
We may assume that $P_i(\a), Q_i(\a) \in \R_{\ge0}$ for all $i \in [n]$.
Because $|\l P, Q\r| = \prod_{i=1}^n |\l P_i, Q_i\r|$, it suffices to argue that for each $i$,
\[ 
|\l P_i, Q_i\r| \le 1 - \Omega(\|P_i-Q_i\|_2^2). 
\]
Fix $i$ and let $\theta\in\mathbb{C}$ be of absolute value $1$ such that $|\langle P_i, Q_i\rangle|=\theta \langle P_i, Q_i\rangle$. Then
\[
|\langle P_i, Q_i\rangle|
=\Re(\langle \theta P_i, Q_i\rangle)
=\frac{1}{2}\E_x\left[|P_i(x)|^2 + |Q_i(x)|^2
-|\theta P_i(x) - Q_i(x)|^2\right].
\]
Denote $\tau_i=\|P_i-Q_i\|_2^2$. If $|P_i(x)|\leq 1-\frac{1}{8}\tau_i$ or $|Q_i(x)|\leq 1-\frac{1}{8}\tau_i$
for some $x$, then (using the $1$-boundedness of $P_i,Q_i$) we immediately get from the previous inequality that $|\langle P_i, Q_i\rangle|\leq 1-\Omega(\tau_i)$, and we are done. Hence, we assume henceforth that $|P_i(x)|,|Q_i(x)|>1-\frac{1}{8}\tau_i$ for all $x$. We get that
\[
|\langle P_i, Q_i\rangle|
\leq 1-\frac{1}{2}\E_x\left[|\theta P_i(x) - Q_i(x)|^2\right].
\]
In the rest of the argument we show that $\E_x\left[|\theta P_i(x) - Q_i(x)|^2\right]\geq \Omega(\tau_i)$ by case analysis. If $|\theta - 1|\leq \frac{1}{2}\sqrt{\tau_i}$, then by the triangle inequality
\[
\E_x\left[|\theta P_i(x) - Q_i(x)|^2\right]^{1/2}
\geq 
\E_x\left[|P_i(x) - Q_i(x)|^2\right]^{1/2}
-|1-\theta|\E_x\left[|P_i(x)|^2\right]^{1/2}
\geq \frac{1}{2}\sqrt{\tau_i},
\]
and so 
$\E_x\left[|\theta P_i(x) - Q_i(x)|^2\right]\geq \frac{1}{4}\tau_i$ and we are done.
Else, $|\theta - 1| > \frac{1}{2}\sqrt{\tau_i}$, 
and then
\[
\E_x\left[|\theta P_i(x) - Q_i(x)|^2\right]
\geq \Omega\left(|\theta P_i(\sigma) - Q_i(\sigma)|^2\right)
\geq 
|P_i(\sigma)\Re(\theta) - Q_i(\sigma)|^2
+
|P_i(\sigma)\Im(\theta)|^2.
\]
If $|\Im(\theta)|\geq \frac{1}{4}\sqrt{\tau_i}$ 
then the second term on the right hand side is already $\frac{1}{16}\tau_i$, and we are done.
Else, $|\Im(\theta)|\leq \frac{1}{4}\sqrt{\tau_i}$
and $|\theta - 1| > \frac{1}{2}\sqrt{\tau_i}$, 
so it follows that 
$\Re(\theta)\leq 1 - \frac{1}{4}\sqrt{\tau_i}$.
If $\Re(\theta)\leq 0$ then
$|P_i(\sigma)\Re(\theta) - Q_i(\sigma)|
\geq |Q_i(\sigma)|\geq 1/2$, and we are done. 
Else, 
\[
|P_i(\sigma)\Re(\theta) - Q_i(\sigma)|
\geq 
|Q_i(\sigma)|
-|P_i(\sigma)\Re(\theta)|
\geq 
1-\frac{1}{8}\tau_i
-
\left(1 - \frac{1}{4}\sqrt{\tau_i}\right)\cdot 1
= \frac{1}{8}\sqrt{\tau_i}.
\qedhere
\]
\end{proof}

We will require a partial converse to~\cref{lemma:prodip}. It essentially says that if $\|\pi(P) - \pi(Q)\|_2^2 \le \eps$, then $\|P - Q\|_2^2$ is small up to some transformations.
\begin{lemma}
\label{lemma:prodip2}
If $P, Q \in \mc{F}([n])$ then there is some $|c| = 1$ such that $\|cP - Q\|_2^2 \le \|\pi(P) - \pi(Q)\|_2^2$.
\end{lemma}
\begin{proof}
Let $|c_i|=1$ be such that $c_i\l P_i, Q_i \r \in \R_{\ge0}$, and let $c = \prod c_i$. Note that
\[ \|P_i - Q_i\|_2^2 = 2 - \l P_i, Q_i \r - \l Q_i, P_i \r, \] and hence $|\l P_i, Q_i \r| \ge 1 - \|P_i - Q_i\|_2^2/2$.
\begin{align*}
    \|cP - Q\|_2^2 &= 2 - \E_x \prod_{i=1}^n c_i P_i(x_i)\bar{Q_i(x_i)} - \E_x \prod_{i=1}^n \bar{c_i} \bar{P_i(x_i)}Q_i(x_i) \\
    &= 2 - 2\prod_{i=1}^n |\l P_i, Q_i\r| \le 2 - 2 \prod_{i=1}^n (1 - \|P_i - Q_i\|_2^2/2) \\
    &\le \sum_{i=1}^n \|P_i - Q_i\|_2^2 = \|\pi(P) - \pi(Q)\|_2^2,
\end{align*}
as desired.
\end{proof}
Now we will describe a procedure which takes a function $f$ and produces a small list of product functions which essentially capture the set of all product functions with nontrivial correlation to $f$. We begin by defining the list of all product functions that are correlated with $f$.
\begin{definition}
\label{def:list}
Given a $1$-bounded function $f: \A^n \to \bbC$, define $\lst_{\eps}[f]$ to be the set of product functions $P \in \mc{F}([n])$ with $|\l f, P\r| \ge \eps$.
\end{definition}

The next result extracts from $\lst_{\eps}[f]$ a short list that is essentially a net.
\begin{lemma}[\!\!{\cite[Lemma 12.16]{BKM4}}]
\label{lemma:slist}
Let $f: \A^n \to \bbC$ be a $1$-bounded function. There is some $\slst_{\eps,\delta}[f] \subseteq \lst_{\eps}[f]$ satisfying:
\begin{enumerate}
    \item $|\slst_{\eps,\delta}[f]| \le \frac{1}{\eps^2-\delta}$, and
    \item For all $P \in \lst_{\eps}[f]$ there is $P' \in \slst_{\eps,\delta}[f]$ wuth $|\l P, P'\r| \ge \delta$.
\end{enumerate}
\end{lemma}

We will require a lemma which says that if $f$ correlates to a product function, then it does so with nonnegligible probability under random restriction.
\begin{lemma}
\label{lemma:rrprod}
Let $f: \A^n \to \bbC$ be a $1$-bounded function, and $P$ be a $1$-bounded product function with $|\l f, P \r| \ge \eps$. For any subset $I \subseteq [n]$, it holds that
\[ \Pr_{z \sim \A^{\bar{I}}}\left[|\l f_{\bar{I} \to z}, P_{\bar{I} \to z}\r| \ge \eps/2\right] \ge \eps/2. \]
\end{lemma}
\begin{proof}
Follows by a simple averaging argument, as $f, P$ are $1$-bounded.
\end{proof}

\subsection{From Random Restrictions to a Direct Product Test}
\label{subsec:todp}
We will require the following theorem which turns local agreement into global agreement. If $D = 0$, it would follow by standard small-set-expansion arguments.
\begin{restatable}{theorem}{sse}
\label{thm:sseproj}
Let $M, n$ be positive integers and $\rho, \gamma \in (0, 1)$. Let $f: \{0, 1\}^n \to \R^M$ be such that:
\[ \Pr_{x \sim_{\rho} [n], y \sim_{1-\gamma} x}[\|f(x) - f(y)\|_2 \le D] \ge \eps. \]
Then
\[ \Pr_{x,y \sim_{\rho} [n]}[\|f(x) - f(y)\|_2 \le O_{\rho,\gamma}(D \log(1/\eps))] \ge \eps^{O_{\rho,\gamma}(1)}. \]
\end{restatable}
\begin{proof}
Deferred to~\cref{subsec:sse}.
\end{proof}
The hypothesis of \cref{thm:rit} guarantees that a random restriction of $f$ correlates to a product function on $\A^n$. We can think of the product function that correlates with $f_{I \to z}$ as an element of $(\R^{2|\A|})^I$. Consider $F: \{0, 1\}^n \to (\R^{2|\A|})^{\le n}$ as the map from $I$ to the correlating product function (later we argue that the product function correlating to $f_{I \to z}$ does not heavily depend on $z$). We will prove that this function $F[I]$ satisfies the following \emph{direct product test}.
\begin{definition}[Direct product test]
\label{def:dp}
For $\rho, \alpha > 0$, we say that a function $F: \{0, 1\}^n \to (\R^K)^{\le n}$ for integer $K$ and parameter $D$ passes the $\DP(\rho, \alpha)$ with probability:
\[ \Pr_{\substack{C \sim_{\alpha\rho} [n] \\ A, B \sim_{\rho} [n] : A, B \supseteq C}} \left[\|F[A]|_C - F[B]|_C\|_2 \le D \right]. \]
Here, the probability is over $C$ sampled from $\sim_{\rho} [n]$ and $A, B$ sampled from $\sim_{\rho} [n]$ conditioned on $A \supseteq C$ and $B \supseteq C$.
\end{definition}
Then, we prove the following inverse theorem for the direct product test of \cref{def:dp}.

\begin{restatable}{theorem}{directproduct}
\label{thm:dp}
Let $n, K$ be positive integers such that $F: \{0, 1\}^n \to ([0, 1]^K)^{\le n}$ with parameter $D > 1$ passes $\DP(\rho, \alpha)$ with probability at least $\eps$. Then there is a function $g: [n] \to [0, 1]^K$ with
\[ \Pr_{A \sim_{\rho} [n]}\left[\|g|_A - F[A]\|_2 \le O_{\rho,\alpha}(D\log(1/\eps) + \log(1/\eps)^{3/2}) \right] \ge \eps^{O_{\rho,\alpha}(1)}. \]
Also, for each $x \in [n]$ there is some $A \subseteq [n]$ such that $g(x) = F[A]_x$.
\end{restatable}

The remainder of the section contains the proof of \cref{thm:rit} given \cref{thm:sseproj,thm:dp}.
\subsubsection{Local structure}
Note that the short lists of $f_{\bar{I} \to z}$ (as in \cref{lemma:slist}) may depend on both $I$ and $z$. Our next goal will be to define short lists that do not depend on $z$ and use these to define a direct product test.

Let $K > C_{\rho}\log(1/\eps)^{O(1)}$ and $\zeta < \eps^{O_{\rho}(1)}$ for sufficiently large constant $C_{\rho}$. For $I \subseteq [n]$ define
\[ W_I := \{P \in \mc{F}(I) : \Pr_{z \in \A^{\bar{I}}}\left[\exists Q \in \slst_{\eps/2, \eps^2/10}[f_{\bar{I} \to z}], \|\pi(P)-\pi(Q)\|_2^2 \le K \right] \ge \zeta \}, \]
i.e., $W_I$ is the set of product functions that are close to something in $\slst_{\eps/2, \eps^2/10}[f_{\bar{I} \to z}]$ for a noticeable fraction of assignments $z$ to $\bar{I}$. The next lemma produces a short list capturing $W_I$.
\begin{lemma}
\label{lemma:sw}
There is a subset $SW_I \subseteq W_I$ satisfying:
\begin{enumerate}
    \item $|SW_I| \le O(\eps^{-2}\zeta^{-1})$.
    \item For all $P \in W_I$ there is $Q \in SW_I$ with $\|\pi(P)-\pi(Q)\|_2^2 \le 4K$.
\end{enumerate}
\end{lemma}
\begin{proof}
For $Q \in \mc{F}(I)$ define \[ S(Q) := \{ (z, P) : P \in \slst_{\eps/2, \eps^2/10}[f_{\bar{I} \to z}], \|\pi(P)-\pi(Q)\|_2^2 \le K \}. \]
Let $SW_I$ consist of a maximal set of $Q \in W_I$ such that $S(Q)$ are disjoint. We first prove that item 2 is satisfied. Consider a $P \in W_I$ where $(z, Q') \in S(P) \cap S(Q)$ for some $Q \in SW_I$. Then $\|\pi(P) - \pi(Q)\|_2^2 \le 2\|\pi(P) - \pi(Q')\|_2^2 + 2\|\pi(Q') - \pi(Q)\|_2^2 \le 4K$, by the triangle inequality.

Towards item 1, for a subset $S \subseteq \A^{\bar{I}} \times \mc{F}(I)$ define its \emph{index} as $\mc{I}(S) = \sum_{(z, P) \in S} \mu(z)$, where $\mu$ is the measure over $z \in \A^{\bar{I}}$. Then
\[ \zeta|SW_I| \le \mc{I}(\cup_{Q \in SW_I} S(Q)) \le O(\eps^{-2}), \]
where the first inequality is by disjointness and the definition of $S(Q)$, and the second inequality is by \cref{lemma:slist}. Indeed, note that if $(z, P) \in S(Q)$ then $P \in \slst_{\eps/2, \eps^2/10}[f_{\bar{I} \to z}]$, and each short list has size at most $O(\eps^{-2})$.
\end{proof}
Our next goal is to argue that for many $I' \subseteq_{1/2} I \subseteq_{\rho} [n]$ it holds that $SW_{I'}$ is nonempty. This is by a small-set-expansion argument.
\begin{lemma}
\label{lemma:swi}
It holds that
\[ \Pr_{I' \subseteq_{\rho/2} [n]}\left[|SW_{I'}| > 0 \right] \ge \Omega(\eps^7). \]
\end{lemma}
\begin{proof}
For each $z \in \A^{\bar{I'}}$, let $Q_z \in \slst_{\eps/2, \eps^2/10}[f_{\bar{I'} \to z}]$ be random. We argue that
\begin{align} \Pr_{\substack{I \subseteq_{\rho} [n], z \sim \A^{\bar{I}} \\ I' \subseteq_{1/2} I, z', z'' \sim \A^{I \setminus I'}}}\left[\|\pi(Q_{(z,z')}) - \pi(Q_{(z,z'')})\|_2^2 \le O(\log(1/\eps)) \right] \ge \Omega(\eps^7). \label{eq:localstruct} \end{align}
Indeed, by the hypothesis in \cref{thm:rit}, with probability $\eps$ over $(I, z)$ there is some $P$ with $|\l f_{\bar{I} \to z}, P \r| \ge \eps$. By \cref{lemma:rrprod}, $|\l f_{\bar{I'} \to (z, z')}, P|_{I'} \r| \ge \eps/2$ with probability at least $\eps/2$ over $z'$, and thus there is some $Q \in \slst_{\eps/2, \eps^2/10}[f_{\bar{I'} \to (z,z')}]$ with $|\l P|_{I'}, Q \r| \ge \eps^2/10$, and hence $\|\pi(P|_{I'}) - \pi(Q)\|_2^2 \le O(\log(1/\eps))$ by \cref{lemma:prodip}.

Hence with probability at least $(\eps/2)^2$ conditioned on $(I, z)$, there are \[ Q' \in \slst_{\eps/2, \eps^2/10}[f_{\bar{I'} \to (z,z')}] \enspace \text{ and } \enspace Q'' \in \slst_{\eps/2, \eps^2/10}[f_{\bar{I'} \to (z,z'')}] \] with $\|\pi(P|_{I'}) - \pi(Q')\|_2^2 \le O(\log(1/\eps))$ and $\|\pi(P|_{I'}) - \pi(Q'')\|_2^2 \le O(\log(1/\eps))$, so $\|\pi(Q') - \pi(Q'')\|_2^2 \le O(\log(1/\eps))$ by the triangle inequality. The probability that these specific $Q', Q''$ are picked from the short lists is $\Omega(\eps^4)$, by \cref{lemma:slist}. Thus \eqref{eq:localstruct} follows.

Now for $K = C\log(1/\eps)^{O(1)}$ and small $\zeta = \eps^{O_{\rho}(1)}$, the lemma follows from \cref{thm:sseproj}.
\end{proof}

Now we define a relaxed version of $W_I$. Define
\[ \wt{W}_I := \{P \in \mc{F}(I) : \Pr_{z \in \A^{\bar{I}}}\left[\exists Q \in \slst_{\eps/4, \eps^2/100}[f_{\bar{I} \to z}], \|\pi(P) - \pi(Q)\|_2^2 \le O(K+\log(1/\eps)) \right] \ge \zeta\eps/4 \}. \]
We can define $\wt{SW}_I \subseteq \wt{W}_I$ as in \cref{lemma:sw}.
Finally we argue that if $|SW_{I'}| > 0$ and $P \in SW_{I'}$ then for any $I'' \subseteq I'$ it holds that $P|_{I''} \in \wt{SW}_{I''}$.
\begin{lemma}
\label{lemma:swii}
Let $P \in SW_{I'}$ and $I'' \subseteq I$. Then $P|_{I''} \in \wt{SW}_{I''}$.
\end{lemma}
\begin{proof}
Let $z \sim \A^{\bar{I'}}$ and $Q \in \slst_{\eps/2, \eps^2/10}[f_{\bar{I'} \to z}]$ be such that $\|\pi(P)-\pi(Q)\|_2^2 \le K$. Such $Q$ exists with probability at least $\zeta$ over $z$ because $P \in SW_{I'}$. By \cref{lemma:rrprod} we know that
\[ \Pr_{z'\sim \A^{I' \setminus I''}}\left[|\l f_{\bar{I''} \to (z,z')}, Q|_{I''}\r| \ge \eps/4 \right] \ge \eps/4. \] Thus with probability at least $\zeta\eps/4$ over $(z, z')$, $Q|_{I''} \in \lst_{\eps/4}[f_{\bar{I''} \to (z,z')}]$. By the definition of $\slst_{\eps/4, \eps^2/100}$ and \cref{lemma:prodip}, we know that there is $P' \in \slst_{\eps/4, \eps^2/100}$ with $\|\pi(P') - \pi(Q|_{I''})\|_2^2 \le O(\log(1/\eps))$. Thus
\[ \|\pi(P') - \pi(P|_{I''})\|_2^2 \le 2\|\pi(P') - \pi(Q|_{I''})\|_2^2 + 2\|\pi(P) - \pi(Q)\|_2^2 \le O(K + \log(1/\eps)). \]
This holds with probability at least $\zeta\eps/4$ over $(z, z')$ as desired.
\end{proof}

\subsubsection{Applying the Direct Product Test}

Define a function $F: \{0,1\}^n \to (\R^{2|\A|})^{\le n}$ as follows. Let $I \subseteq [n]$ and define $F[I]$ to be a uniformly random element of $\wt{SW}_I$ (and apply $\pi$) if it is nonempty, and otherwise a random point.
The following lemma bounds the success probability of the direct product test.
\begin{lemma}
\label{lemma:dplower}
It holds that
\[ \Pr_{\substack{I' \subseteq_{\rho/2} [n] \\ I'', I''' \subseteq_{1/2} I}}\left[\|F[I'']_{I'' \cap I'''} - F[I''']_{I'' \cap I'''}\|_2^2 \le O(K + \log(1/\eps)) \right] \ge \Omega(\zeta^2\eps^{11}). \]
\end{lemma}
\begin{proof}
By \cref{lemma:swi} we know that $|SW_{I'}| > 0$ with probability at least $\Omega(\eps^7)$. Fix $P \in SW_{I'}$. By \cref{lemma:swii} we know that $P|_{I''} \in \wt{SW}_{I''}, \wt{SW}_{I'''}$. The result follows from $|\wt{SW}_{I''}|, |\wt{SW}_{I'''}| \le O(\zeta^{-1}\eps^{-2})$ by \cref{lemma:sw} and the triangle inequality.
\end{proof}
We now have all the pieces necessary to prove \cref{thm:rit}.
\begin{proof}[Proof of \cref{thm:rit}]
By \cref{lemma:dplower} and \cref{thm:dp} there is a $1$-bounded global function $g: [n] \to \bbC^{|\A|}$ and constants $K' \le O_{\rho}((K+\log(1/\eps))\log(1/\eps)^{O(1)})$ and $\eps' \ge (\eps\zeta)^{O_{\rho}(1)}$ such that
\begin{align} \Pr_{I \subseteq_{\rho/4} [n]}\left[\|g|_I - F[I]\|_2^2 \le K' \right] \ge \eps'. \label{eq:dpresult}
\end{align}
Let $G(x) = \prod_{i=1}^n g(i)_{x_i}$ be the product function corresponding to $g$. Let $\gamma = \eps^3/(CK')$ for sufficiently large constant $C$. 

By \eqref{eq:dpresult} and the definition of $\wt{SW}_I$, we know that
\[ \Pr_{I \subseteq_{\rho/4} [n], z \sim \A^{\bar{I}}}\left[\exists Q \in \mc{F}(I),|\l f_{\bar{I} \to z}, Q \r| \ge \eps/2, \|\pi(G|_I) - \pi(Q)\|_2^2 \le O(K') \right] \ge \eps'\zeta. \]
For this fixed pair $(I, z)$ we conclude that
\begin{align} \Pr_{I' \subseteq_{\gamma} I, z' \sim \A^{I \setminus I'}}\left[|\l f_{\bar{I}' \to (z,z')}, Q_{(I\setminus I') \to z'} \r| \ge \eps/4, \|\pi(G|_{I'}) - \pi(Q_{(I \setminus I') \to z'})\|_2^2 \le c\eps^2 \right] \ge \eps/8. \label{eq:fg}  \end{align}
This follows because $|\l f_{\bar{I}' \to (z,z')}, Q_{(I\setminus I') \to z'} \r| \ge \eps/4$ with probability at least $\eps/4$ by \cref{lemma:rrprod}, and also, $\E[\|G|_{I'} - Q_{(I \setminus I') \to z'}\|_2^2] \le O(\gamma K')$, and hence
\[ \Pr[\|\pi(G|_{I'}) - \pi(Q_{(I \setminus I') \to z'})\|_2^2 > c\eps^2] \le O(\gamma K'/(c\eps^2)) \le \eps/8, \] by Markov's inequality and the choice of $\gamma$. Let $|\alpha| = 1$ be such that
\[ \|\alpha G|_{I'} - Q_{(I \setminus I') \to z'}\|_2^2 \le c\eps^2, \] which exists by \cref{lemma:prodip2}. Thus
\begin{align*}
\stab_0((f\bar{G})_{\bar{I} \to (z,z')}) &= |\l f_{\bar{I} \to (z,z')}, \alpha G|_{I'}\r|^2 \\ 
&\ge \left(|\l f_{\bar{I}' \to (z,z')}, Q_{(I\setminus I') \to z'} \r| - \|\alpha G|_{I'} - Q_{(I \setminus I') \to z'}\|_2 \right)^2\\ 
&\ge \frac{\eps^2}{64}.
\end{align*}
Thus
\begin{align*}
    \stab_{1-\rho\gamma/4}(f\bar{G}) &= \E_{\substack{I \subseteq_{\rho/4} [n], I' \subseteq_{\gamma} I \\ z \sim \A^{\bar{I}}, z' \sim \A^{I \setminus I'}}}\left[\stab_0((f\bar{G})_{\bar{I} \to (z,z')}) \right]
    \ge \eps'\zeta \cdot \frac{\eps}{8} \cdot \frac{\eps^2}{64} \ge \frac{\eps'\zeta\eps^3}{512}.
\end{align*}
Thus \cref{thm:rit} follows for degree $D := O(\rho\gamma \log(1/\delta'))$ for $\delta' := \frac{\eps'\zeta\eps^3}{2000}.$
\end{proof}

\subsection{Small-Set Expansion}
\label{subsec:sse}
The goal of this section is to establish \cref{thm:sseproj}.
We require the following standard result on small-set expansion in the biased hypercube.
\begin{theorem}[Small-set expansion]
\label{thm:sse}
Let $\rho, \gamma \in (0, 1)$, and $A \subseteq \{0, 1\}^n$. Let $\mu(A) := \Pr_{x \sim_\rho [n]}[x \in A]$. There is a constant $c = \Omega_{\rho,\gamma}(1)$ such that for all $A$, $\Pr_{x \sim_{\rho} [n], y \sim_{1-\gamma}(x)}[x, y \in A] \le \mu(A)^{1+c}$.
\end{theorem}
\begin{proof}
Let $T_{\gamma}\colon L_2(\{0,1\}^n; \mu_\rho^{\otimes n})\to L_2(\{0,1\}^n; \mu_\rho^{\otimes n})$ be the averaging operator 
defined as $T_{\gamma} f(x) = \E_{y\sim_{1-\gamma}(x)}{f(y)}$. Then 
\[
\Pr_{x \sim_{\rho} [n], y \sim_{1-\gamma}(x)}[x, y \in A] 
=\langle 1_A, T_{\gamma} 1_A\rangle
\leq \|1_A\|_2\|T_{\gamma} 1_A\|_2.
\]
By~\cite[Chapter 10]{O14} 
we have that $\|T_{\gamma} 1_A\|_2\leq \|1_A\|_{2-c}$
for some $c = c(\rho,\gamma)>0$, giving us that 
\[
\Pr_{x \sim_{\rho} [n], y \sim_{1-\gamma}(x)}[x, y \in A] 
\leq 
\|1_A\|_2
\|1_A\|_{2-c}
=\mu(A)^{1/2}\cdot\mu(A)^{1/(2-c)}
\leq \mu(A)^{1+c/4}.\qedhere
\]
\end{proof}
Our first goal is to apply \cref{thm:sse} to establish an analogue of \cref{thm:sseproj} in the setting where the dimension $m$ is small.
\begin{lemma}
\label{lemma:constm}
Let $m, n$ be positive integers and $\rho, \gamma \in (0, 1)$. Let $f: \{0, 1\}^n \to \R^m$ be such that:
\[ \Pr_{x \sim_{\rho} [n], y \sim_{1-\gamma} x}[\|f(x) - f(y)\|_2 \le D] \ge \eps. \]
Then
\[ \Pr_{x,y \sim_{\rho} [n]}[\|f(x) - f(y)\|_2 \le 2Dm] \ge (\eps/2)^{O_{\rho,\gamma}(1)}.\]
\end{lemma}
\begin{proof}
For all $v \in \R^m$, let
\[ S_v = \{x \in \{0,1\}^n: f(x)_i \in [v_i, v_i + 2D\sqrt{m}] \enspace \text{ for all } \enspace i \in [m]\}. \] Note that each $x \in \R^n$, $\int_{v \in \R^m} 1_{x \in S_v} = (2D\sqrt{m})^m$. Also, for each $x, y \in \R^m$ with $\|x-y\|_2 \le D$ it holds that
\[ \int_{v \in \R^m} 1_{x,y \in S_v} = \prod_{i=1}^m (2D\sqrt{m} - |x_i-y_i|) \ge (2D\sqrt{m})^m\Big(1 - \sum_{i=1}^m \frac{|x_i-y_i|}{2D\sqrt{m}}\Big) \ge (2D\sqrt{m})^m/2. \]
Combining these along with small-set-expansion (\cref{thm:sse})
\begin{align*}
    \eps 
    &\le \Pr_{x \sim_{\rho} [n], y \sim_{1-\gamma} x}[\|f(x) - f(y)\|_2 \le D]\\
    &\le \frac{2}{(2D\sqrt{m})^m} \E_{x \sim_{\rho} [n], y \sim_{1-\gamma} x}\left[\int_{v \in \R^m} 1_{x,y \in S_v}\right] \\
    &\le \frac{2}{(2D\sqrt{m})^m} \int_{v \in \R^m} \mu(S_v)^{1+c} \\
    &\le 2 \max_v \mu(S_v)^c,
\end{align*}
where we have applied \cref{thm:sse}. Thus, $\mu(S_v) \ge (\eps/2)^{O_{\rho,\nu}(1)}$ for some $v \in \Z^m$. To conclude, note for any $x, y \in S_v$ that $\|f(x) - f(y)\|_2 \le 2Dm$ and
\[ \Pr_{x,y \sim_{\rho} [n]}[\|f(x) - f(y)\|_2 \le 2Dm] \ge \mu(S_v)^2 \ge (\eps/2)^{O_{\rho,\gamma}(1)}. \qedhere \]
\end{proof}
Now to prove \cref{thm:sseproj} we will essentially randomly project the values of $f(x)$ down to a constant number of dimensions, and then apply \cref{lemma:constm}. Towards this we first prove that a random projection indeed preserves the $\ell_2$ norm of pairwise distances in a way that we need.

\begin{lemma}
\label{lemma:proj}
Let $u, v \in \R^M$ and let $g_1, \dots, g_m \sim \mathcal{N}(0, I)$ be uniform Gaussian vectors in $M$ dimensions. Then for $u' = m^{-1/2}(\l u, g_1 \r, \dots, \l u, g_m \r) \in \R^m$ and $v' = m^{-1/2}(\l v, g_1 \r, \dots, \l v, g_m \r)$ it holds that
\[ \Pr[\|u'-v'\|_2^2 > 2\|u-v\|_2^2] \le e^{-m/10} \enspace \text{ and } \enspace \Pr[\|u'-v'\|_2^2 < \delta\|u-v\|_2^2] \le (10\delta)^{m/2}.\]
\end{lemma}
\begin{proof}
We may assume without loss of generality that $v = 0$ and that $\|u\|_2=1$. The first bound then follows from the Hanson-Wright inequality, because each $\l u, g_i\r$ is Gaussian with variance $\|u\|_2^2 = 1$. For the second bound let $\eta_i = \l u, g_i \r/\|u\|_2$, and note that
\begin{align*} 
\Pr[\|u'\|_2^2 \le \delta\|u\|_2^2] &= \int_{\sum_{i=1}^m \eta_i^2 \le \delta m} \prod_{i=1}^m \frac{1}{\sqrt{2\pi}}e^{-\eta_i^2/2} \le (2\pi)^{-m/2} \mathrm{Vol}(\sqrt{\delta m} B_m) \le (10\delta)^{m/2},
\end{align*}
where $B_m$ is the unit ball in $m$ dimensions.
\end{proof}
At this point, we are ready to proceed to the proof of \cref{thm:sseproj}.
\begin{proof}[Proof of \cref{thm:sseproj}]
Let $m, \delta$ be parameters chosen later. By Markov's inequality and \cref{lemma:proj} there are $g_1, \dots, g_m \in \R^M$ such that $g(u) = m^{-1/2}(\l u, g_1\r, \dots, \l u, g_m\r)$ satisfies:
\begin{equation} \Pr_{x,y \sim_{\rho} [n]}[\|g(u)-g(v)\|_2 < \delta^{1/2}\|f(u)-f(v)\|_2] \le 3(10\delta)^{m/2} \label{eq:proj1} \end{equation} and
\[ \Pr_{x \sim_{\rho} [n], y \sim_{1-\gamma} x}[\|g(u)-g(v)\|_2 \le 2D] \ge \eps - 3e^{-m/10} \ge \eps/2 \] for $m = 100C_{\rho,\gamma}\log(1/\eps)$ for sufficiently large constant $C_{\rho,\gamma} = O_{\rho,\gamma}(1)$. Combining the latter equation with \cref{lemma:constm} gives
\begin{equation} \Pr_{x,y \sim_{\rho} [n]}[\|g(x) - g(y)\|_2 \le 4Dm] \ge (\eps/2)^{O_{\rho,\gamma}(1)} \ge \eps^{O_{\rho,\gamma}(1)}. \label{eq:proj3} \end{equation}
Let $\delta$ be sufficiently small so that $6(10\delta)^{m/2} < \eps^{O_{\rho,\gamma}(1)}$, here the constant $C_{\rho,\gamma}$ in the definition of $m$ is sufficiently large. 
Combining~\eqref{eq:proj3} with \eqref{eq:proj1} gives
\[ \Pr_{x,y \sim_{\rho} [n]}[\|f(u)-f(v)\|_2 \le 4\delta^{-1/2}Dm] \ge \eps^{O_{\rho,\gamma}(1)} - 3(10\delta)^{m/2} \ge \eps^{O_{\rho,\gamma}(1)}, \] as desired.
\end{proof}

\subsection{Analysis of Direct Product Test}
\label{subsec:dp}

In this section we prove \cref{thm:dp}, restated below.
\directproduct*
We first rephrase the distribution in the direct product test. Note that we may find constants $\gamma, \gamma' \in (0, 1)$ such that the distribution over $A, B, C$ in \cref{def:dp} is equivalent to sampling $A \sim_{\rho} [n]$, $B \sim_{1-\gamma} A$, and $C \subseteq_{\gamma'} A \cap B$. Here $B \sim_{1-\gamma} A$ means that $A, B$ are $1-\gamma$ correlated.

Let us work with this formulation from now on. Our first goal is to establish that $\|F[A]|_{A \cap B} - F[B]|_{A \cap B}\|_2$ is small.
\begin{lemma}
\label{lemma:dp1}
Under the hypotheses of \cref{thm:dp} it holds that
\[ \Pr_{A \sim_{\rho} [n], B \sim_{1-\gamma} A} \left[\|F(A)|_{A \cap B} - F[B]_{A \cap B}\|_2 \le O_{\rho,\gamma}(D + \log(1/\eps)^{1/2}) \right] \ge \frac{2\eps}{3}. \]
\end{lemma}
\begin{proof}
Follows by the fact that $C \subseteq_{\gamma'} A \cap B$ for some $\gamma' > 0$ and the Chernoff bound.
\end{proof}
For the remainder of the section, let $D' = O(D + \log(1/\eps)^{1/2})$.
\begin{definition}[Consistent]
\label{def:cons}
We say that $A, B$ are $D'$-consistent if $\|F[A]|_{A \cap B} - F[B]|_{A \cap B}\|_2 \le D'$.
\end{definition}
If $A, B$ are $D'$-consistent we sometimes write that $A \in \cons(B)$ or $B \in \cons(A)$.

\begin{definition}[Good]
\label{def:good}
We say that $A$ is good if $\Pr_{B \sim_{1-\gamma} A}[B \in \cons(A)] \ge \eps/3$.
\end{definition}
By an averaging argument, we know that $\Pr_{A \sim_{\rho} [n]}[A \text{ is good}] \ge \eps/3$. 
Let $r$ be a parameter to be determined.
\begin{definition}[Excellent]
\label{def:excellent}
For some $D'' > D'$, we say that $A$ is excellent if $A$ is good and
\[ \Pr_{B \sim_{1-\gamma} A}\left[\E_{B' \sim_{1-\gamma} A}\left[\|F[B]|_{B \cap B'} - F[B']_{B \cap B'}\|_2^2 \enspace | \enspace B' \in \cons(A) \right] > (D'')^2 \text{ and } B \in \cons(A)\right] \le r. \]
\end{definition}
Now we establish for large enough $D''$ that the probability that $A$ is good but not excellent is exponentially small.
\begin{lemma}
\label{lemma:excellent}
If $D'' > C_{\rho,\gamma}D'$ then
\[ \Pr_{A \sim_{\rho} [n]}[A \emph{ is good but not excellent}] \le \exp(-\Omega_{\rho,\gamma}(D''))/r. \]
\end{lemma}
\begin{proof}
Let $C := C_{\rho,\gamma}$. For a real number $X \ge 0$, let $E(A,B,B',X)$ be the event that
\[ \|F[B]|_{B \cap B'} - F[B']|_{B \cap B'}\|_2^2 \ge C\|F[B]|_{A \cap B \cap B'} - F[B']|_{A \cap B \cap B'}\|_2^2 + X. \]
The main claim is that $\Pr_{A,B,B'}[E] \le \exp(-\Omega_{\rho,\gamma}(X))$. To see this, consider first fixing $B. B'$. Note for $i \in B \cap B'$, there is some probability $p > 0$ (for each coordinate iid, independent of $n$), such that $i \in A$. Thus, the claim follows from Bernstein's inequality.

Note that when $E(A,B,B',X)$ does not hold and $(A, B), (A, B')$ are $D'$-consistent:
\begin{align*} 
\|F[B]|_{B \cap B'} - F[B']|_{B \cap B'}\|_2^2 &\le X + C\|F[B]|_{A \cap B \cap B'} - F[B']|_{A \cap B \cap B'}\|_2^2 \\
&\le X + 2C\|F[B]|_{A \cap B \cap B'} - F[A]|_{A \cap B \cap B'}\|_2^2 \\ ~&+ 2C\|F[A]|_{A \cap B \cap B'} - F[B']|_{A \cap B \cap B'}\|_2^2 \\
&\le X + 2CD + 2CD \le X + O(D).
\end{align*}
where we have used the triangle inequality and consistency of $(A, B)$ and $(A, B')$. As  $\Pr_{A,B,B'}[E] \le \exp(-\Omega_{\rho,\gamma}(X))$, we get that
\begin{align}
\E_{A,B,B'}\left[1_{B,B' \in \cons(A), A \text{ good}} \exp\left(\|F[B]|_{B \cap B'} - F[B']|_{B \cap B'}\|_2^2/C'\right) \right] \le \exp(O_{\gamma,\rho}(D)), \label{eq:tomarkov}
\end{align}
for sufficiently large constant $C'$ depending on $\rho, \gamma$. Consider good $A$ and $B \in \cons(A)$ with
\[ \E_{B'}\left[1_{B' \in \cons(A)} \exp\left(\|F[B]|_{B \cap B'} - F[B']|_{B \cap B'}\|_2^2/C'\right) \right] \le \exp(O_{\gamma,\rho}(D)). \]
Then \[ \E_{B'}\left[1_{B' \in \cons(A)} \|F[B]|_{B \cap B'} - F[B']|_{B \cap B'}\|_2^2 \right] \le O_{\gamma,\rho}(D) \] follows by convexity of $\exp$. The lemma follows from Markov's inequality applied to \eqref{eq:tomarkov}.
\end{proof}

Now we define how to generate a global function $g_A$ for each set $A \subseteq [n]$. Ultimately, our goal is to prove that these $g_A$ agree for a nonnegligible fraction of $A$.
\begin{definition}[Voting]
\label{def:vote}
Define the function $g_A: [n] \to [0, 1]^K$ as
\[ g_A(x) = \min_{v = F[B']_x : B' \in \{0,1\}^n} \E_{B \sim_{1-\gamma} A}\Big[1_{B \in \cons(A), x \in B}\|v-F[B]_x\|_2^2\Big]. \]
\end{definition}
In other words, $g_A(x)$ is the minimizer over all $F[B]_x$ of the expected distance between it and a random $B$ that is consistent with $A$.

We require a lemma on the conditions of product distributions.
\begin{lemma}
\label{lemma:prodcond}
Let $U = (U_1, \dots, U_n)$ be a product distribution. Let $\wt{U}$ be $U$ conditioned on an event with probability at least $d$ and $\wt{U}_1, \dots, \wt{U}_n$ be the marginal distributions of $\wt{U}$. Then \[ \sum_{i=1}^n \dtv(U_i, \wt{U}_i)^2 \le \log(1/d). \]
\end{lemma}
\begin{proof}
Use the following sequence of inequalities:
\[ \log(1/d) \ge \dkl(\wt{U} \| U) \ge \sum_{i=1}^n \dkl(\wt{U}_i, U_i) \ge \sum_{i=1} \dtv(\wt{U}_i, U_i)^2. \]
The first is because $\wt{U}$ is $U$ conditioned on an event of probability at most $d$, the second is because $U$ is a product distribution, and the last is by Pinsker's inequality.
\end{proof}

For the remainder of the analysis it will be useful to define the \emph{degree} of a set $A$.
\begin{definition}[Degree]
\label{def:degree}
Define $\deg(A) := \Pr_{A' \sim_{1-\gamma} A}\left[A' \in \cons(A) \right]$.
\end{definition}

We first prove that for $A$ with nonnegligible degree that the votes are consistent with $F[A]$.
\begin{lemma}
\label{lemma:votecons}
For all $A$ it holds that $\|(g_A)|_A - F[A]\|_2 \le O_{\rho,\gamma}(D'' + \log(1/\deg(A))^{1/2})$.
\end{lemma}
\begin{proof}
By the definition of consistency we know that
\begin{align*}
    (D')^2 &\ge \E_B[\|F[B]_{A \cap B} - F[A]|_{A \cap B}\|_2^2 \enspace | \enspace B \in \cons(A)] \\
    &= \sum_{x \in A} \E[1_{x \in B} \|F[B]_x - F[A]_x\|_2^2 \enspace | \enspace B \in \cons(A)] \\
    &\ge \sum_{x \in A} \E[1_{x \in B} \|F[B]_x - (g_A)_x\|_2^2 \enspace | \enspace B \in \cons(A)]
\end{align*}
where the final inequality is by the optimality of $g_A$. By the triangle inequality we conclude that
\begin{align*}
    4(D')^2 &\ge \sum_{x \in A} \E[1_{x \in B} \|F[A]_x - (g_A)_x\|_2^2 \enspace | \enspace B \in \cons(A)] \\ &= \sum_{x \in A} \Pr[x \in B \enspace | \enspace B \in \cons(A)] \|F[A]_x - (g_A)_x\|_2^2.
\end{align*}
Now, $\Pr[x \in B \enspace | \enspace B \in \cons(A)] \ge \Omega_{\rho,\gamma}(1)$ except for $O_{\rho,\gamma}(\log(1/\deg(A))$ many $x \in A$, by \cref{lemma:prodcond}. Also, $\|F[A]_x - (g_A)_x\|_2^2 \le O(1)$ for all $x$. Thus,
\[ \|F[A] - (g_A)|_A\|_2^2 = \sum_{x \in A} \|F[A]_x - (g_A)_x\|_2^2 \le O_{\rho,\gamma}((D')^2 + \log(1/\deg(A))) \]
as desired.
\end{proof}
Now we prove that if $A$ is excellent, then $F[B]$ agrees with $g_A$ for most $B$ such that $(A, B)$ is $D'$-consistent.
\begin{lemma}
\label{lemma:voteexc}
If $A$ is excellent then
\[ \Pr_{B \sim_{1-\gamma} A}[\|g_A|_{B} - F[B]\|_2 > C_{\rho,\gamma}D'' \text{ and } B \in \cons(A)] \le O(r/\eps). \]
\end{lemma}
\begin{proof}
Fix $B$ satisfying the condition in \cref{def:excellent}. By definition,
\begin{align*}
    (D'')^2 \ge \E_{B' \sim_{1-\gamma} A}\left[\|F[B]|_{B \cap B'} - F[B']|_{B \cap B'}\|_2^2 \enspace | \enspace B' \in \cons(A) \right].
\end{align*}
By the triangle inequality, and optimality of $g_A$,
\begin{align*} &\E_{B' \sim_{1-\gamma} A}\left[\|g_A|_{B \cap B'} - F[B]|_{B \cap B'}\|_2^2 \enspace | \enspace B' \in \cons(A) \right] \\
&\qquad\le  2\E_{B' \sim_{1-\gamma} A}\left[\|F[B]|_{B \cap B'} - F[B']|_{B \cap B'}\|_2^2 + \|g_A|_{B \cap B'} - F[B']|_{B \cap B'}\|_2^2 \enspace | \enspace B' \in \cons(A) \right] \\
&\qquad\le O((D'')^2).
\end{align*}
Rewrite the above equation as
\[ O((D'')^2) \ge \sum_{x \in B} \Pr_{B' \sim_{1-\gamma} A}[x \in B' \enspace | \enspace B' \in \cons(A)] \|(g_A)_x - F[B]_x\|_2^2. \]
To conclude, recall by \cref{lemma:prodcond} that because $\Pr[B' \in \cons(A)] \ge \eps/3$ because $A$ is good, that for all but $O_{\rho,\gamma}(\log(1/\eps))$ many $x \in B$ that
\[ \Pr_{B' \sim_{1-\gamma} A}[x \in B' \enspace | \enspace B' \in \cons(A)] \ge \Omega_{\rho,\gamma}(1). \]
This completes the proof.
\end{proof}

Finally, we prove that for almost all pairs of consistent $A, A'$ that the votes are close in $\ell_2$. The proof closely resembles that of \cref{lemma:excellent}.
\begin{lemma}
\label{lemma:vote}
For $D''' > C_{\rho,\gamma}D''$, We have that
\[ \Pr_{A \sim_{\rho}[n], A' \sim_{1-\gamma} A}\left[\|g_A - g_{A'}\|_2 \le D''' \text{, and } \deg(A),\deg(A') \ge \eps^{O(1)} \right] \ge \eps^{O(1)}. \]
\end{lemma}
\begin{proof}
We first argue that there are many pairs $(A, A')$ that are $D'$-consistent and $\deg(A), \deg(A') \ge \eps^{O(1)}$. Indeed, let $S = \{A : \deg(A) \ge \eps/3 \}$, i.e., the set of good $A$. Then
\begin{align*}
    &\Pr[A' \in \cons(A), A \in S, A' \in S] \\
    &\qquad\ge \Pr[A' \in \cons(A)] - \Pr[A' \in \cons(A), A \notin S] - \Pr[A' \in \cons(A), A' \notin S] \\
    &\qquad\ge \eps - \eps/3 - \eps/3 \ge \eps/3.
\end{align*}
By \cref{lemma:excellent} at least $\eps/4$ of pairs $(A, A')$ also have that both $A, A'$ are excellent as long as $\exp(-\Omega_{\rho,\gamma}(D''))/r \le \eps/100$ (we will choose $r$ later to ensure this).

We now prove the stronger claim that over consistent pairs $(A, A')$, the probability that $B \sim_{1-\gamma} A$ and $B' \sim_{1-\gamma} A'$ are inconsistent is very unlikely. In other words, $B$ and $B'$ vote similarly towards $g_A$ and $g_{A'}$. Formally, we want to prove that:
\begin{align*} 
&\Pr_{\substack{A \sim_{\rho}[n], A' \sim_{1-\gamma} A \\ B \sim_{1-\gamma} A, B' \sim_{1-\gamma} A'}}\left[A', B \in \cons(A), B' \in \cons(A') \text{ and } \|F[B]|_{B \cap B'} - F[B']|_{B \cap B'}\|_2 > D''' \right] \\ 
&\qquad\le \exp(-\Omega_{\rho,\gamma}(D''')).
\end{align*}
Towards proving this, define $E(A,A',B,B')$ to be the event that
\[ \|F[B]|_{B \cap B'} - F[B']|_{B \cap B'}\|_2 > C_{\rho,\gamma}\|F[B]|_{A\cap A' \cap B \cap B'} - F[B']|_{A\cap A' \cap B \cap B'}\|_2 + D'''/2. \]
We first argue that $\Pr_{A,A',B,B'}[E(A,A',B,B')] \le \exp(-\Omega_{\rho,\gamma}(D'''))$. Indeed, fix $B, B'$ and note that for $i \in B, B'$ there is some probability $p > 0$ (independent of $n$) such that $i \in A \cap A'$. Thus this follows by a Chernoff bound. When $E(A,A',B,B')$ does not hold and $(A, A'), (A, B)$, and $(A', B')$ are all consistent pairs:
\begin{align*}
    \|F[B]|_{B \cap B'} - F[B']|_{B \cap B'}\|_2 &\le D'''/2 + C\|F[B]|_{A\cap A' \cap B \cap B'} - F[B']|_{A\cap A' \cap B \cap B'}\|_2 \\
    &\le D'''/2 + C\|F[B]|_{A\cap A' \cap B \cap B'} - F[A]|_{A\cap A' \cap B \cap B'}\|_2 \\
    &~+ C\|F[A]|_{A\cap A' \cap B \cap B'} - F[A']|_{A\cap A' \cap B \cap B'}\|_2 \\
    &~+ C\|F[A'|_{A\cap A' \cap B \cap B'} - F[B']|_{A\cap A' \cap B \cap B'}\|_2 \\
    &\le D'''/2 + CD' + CD' + CD' \le D'''.
\end{align*}
Here we have used the triangle inequality and consistency. Restrict to $(A, A')$ satisfying 
\[ \Pr_{B \sim_{1-\gamma} A, B' \sim_{1-\gamma} A'}\left[A', B \in \cons(A), B' \in \cons(A') \text{ and } \|F[B]|_{B \cap B'} - F[B']|_{B \cap B'}\|_2 > D''' \right] \le r, \]
for $r$ chosen later. At least $\eps/6 -  \exp(-\Omega_{\rho,\gamma}(D'''))/r \ge \eps/12$ fraction of $(A, A')$ satisfy this. Say that $(B, B')$ is a \emph{valid pair} if $\|F[B]|_{B \cap B'} - F[B']|_{B \cap B'}\|_2 \le D'''$, \[ \|g_A|_{B} - F[B]\|_2 \le O_{\rho,\gamma}(D'') \enspace \text{ and } \enspace \|g_{A'}|_{B'} - F[B']\|_2 \le O_{\rho,\gamma}(D''). \]
By \cref{lemma:voteexc} we know that at most $O(r)$ fraction of pairs $(B, B')$ are not valid.

For such $(A, A')$:
\begin{align*}
    (D''')^2 &\ge \E_{B, B'}\left[\|F[B]|_{B \cap B'} - F[B']|_{B \cap B'}\|_2^2 \enspace | \enspace B \in \cons(A), B' \in \cons(A'), (B, B') \text{ valid} \right],
\end{align*}
and by the triangle inequality:
\begin{align*}
    &\E_{B, B'}\left[\|g_A|_{B \cap B'} - g_{A'}|_{B \cap B'}\|_2^2 \enspace | \enspace B \in \cons(A), B' \in \cons(A'), (B, B') \text{ valid} \right] \\
    &\qquad\le  3\E_{B, B'}\Bigg[\|g_A|_B-F[B]\|_2^2+\|g_{A'}|_{B'}-F[B']\|_2^2+\|F[B]_{B \cap B'} - F[B']_{B \cap B'}\|_2^2 \\
     ~& \qquad\qquad\qquad\qquad\enspace \Big| \enspace B \in \cons(A), B' \in \cons(A'), (B, B') \text{ valid} \Bigg]
     \le O_{\rho,\gamma}((D''')^2).
\end{align*}
Rewriting this gives us:
\begin{align*}
    O_{\rho,\gamma}((D''')^2) \ge \sum_{x \in [n]} \Pr_{B,B'}\left[x \in B \cap B' \enspace | \enspace B \in \cons(A), B' \in \cons(A'), (B, B') \text{ valid} \right] \|(g_A)_x - (g_{A'})_x\|_2^2 .
\end{align*}
Because $B \cap B'$ is a product distribution, and the event $B \in \cons(A), B' \in \cons(A'), (B, B')$ has probability at least $\eps^2/9 - O(r) \ge \eps^2/10$, by \cref{lemma:prodcond} we get that
\[ \Pr_{B, B'}\left[x \in B \cap B' \enspace | \enspace B \in \cons(A), B' \in \cons(A'), (B, B') \text{ valid} \right] \ge \Omega_{\rho,\gamma}(1) \] on all but $O_{\rho,\gamma}(\log(1/\eps))$ many $x \in [n]$. This completes the proof.
\end{proof}
Now \cref{thm:dp} immediately follows from combining \cref{lemma:vote} with \cref{thm:sseproj}, when choosing $r = e^{-C_{\rho,\gamma}' D''}$ for appropriately chosen constant $C_{\rho,\gamma}'>0$.
\begin{proof}[Proof of \cref{thm:dp}]
By \cref{lemma:vote,thm:sseproj} there is a function $g: [n] \to [0,1]^K$ such that
\[ \Pr_{A \sim_{\rho} [n]}\left[\|g - g_A\|_2 \le O_{\rho,\gamma}(D''' \log(1/\eps))), \deg(A) \ge \eps^{O(1)} \right] \ge \eps^{O_{\rho,\gamma}(1)}. \]
Here, we may assume that $\deg(A) \ge \eps^{O(1)}$ by restricting $g_A$ only to $A$ with $\deg(A) \ge \eps^{O(1)}$, and then applying \cref{thm:sseproj}. To conclude, note that \cref{lemma:votecons} gives
\[ \|g|_A - F[A]\|_2 \le \|g - g_A\|_2 + \|F[A] - g_A|_A\|_2 \le O_{\rho,\gamma}(D''' \log(1/\eps)) \] for $\deg(A) \ge \eps^{O(1)}$, and recall that $D''' \le O_{\rho,\gamma}(D + \log(1/\eps)^{1/2})$.
\end{proof}

\section{Applications}
\subsection{Restricted $3$-APs over $\F_p^n$ and Generalizations}
\label{sec:3ap}
In this section we establish \cref{thm:3ap} by a density increment argument.
\threeap*
Throughout this section, let $m = |\A|$. Let $\alpha = m^{-n}|A|$ and define $f_A: \A^n \to \R$ as $f_A(x) = A(x) - \alpha$, where we abuse notation to write $A(x)$ to be the indicator function of $A$. Na\"{i}vely, we wish to have a distribution $\mu$ supported on $S$ such that $\mu_x, \mu_y, \mu_z$ are all uniform. Unfortunately this is not always possible. Instead we settle for an approximate version of uniformity which we then random restrict so that $\mu_x$ becomes uniform.

Throughout this section, assume that $\alpha \ge C_m(\log\log\log n)^{-c_m}$. Call a triple $(x,y,z)$ such that $x, y, z$ are not all equal and $(x_i,y_i,z_i) \in S$ for all $i \in [n]$ a \emph{valid triple}.
\begin{lemma}
\label{lemma:notuniform}
Let $\mu$ be the distribution over $\A \times \A \times \A$ with mass $\frac1m - \frac{\delta}{m\sqrt{n}}$ on $(a,a,a)$ for $a \in \A$, and $\frac{1}{|S|-m} \cdot \frac{\delta}{\sqrt{n}}$ on each $\vec{a}$ for $\vec{a} \in S \setminus \{(a,a,a) : a \in \A\}$. If $A$ has no valid triples, then there are $1$-bounded functions $g: \A^n \to \bbC$, $h: \A^n \to \bbC$ such that:
\begin{align} 
&\max\left\{\left|\E_{(x,y,z) \sim \mu^{\otimes n}}[f_A(x)g(y)h(z)]\right|, \left|\E_{(x,y,z) \sim \mu^{\otimes n}}[g(x)f_A(y)h(z)]\right|, \left|\E_{(x,y,z) \sim \mu^{\otimes n}}[g(x)h(y)f_A(z)]\right| \right\} \notag\\
\ge ~& \alpha^3/7 - \exp(-\Omega(\delta\sqrt{n})). \label{eq:a37}
\end{align}
Also, it holds that $\E_{x \sim \mu_x^{\otimes n}}[A(x)] \ge \alpha-2\delta$.
\end{lemma}
\begin{proof}
Note that if $A$ has no valid triples, then \[ \E_{(x,y,z) \sim \mu^{\otimes n}}[A(x)A(y)A(z)] = \Pr_{(x,y,z) \sim \mu^{\otimes n}}[x=y=z] \le \exp(-\Omega(\delta\sqrt{n})). \] Thus the first part follows from expanding $A = \alpha + f_A$. The second part follows because
\[ \E_{x \sim \mu_x^{\otimes n}}[A(x)] \ge \E_{x \sim \A^n}[A(x)] - \dtv(\mu_x^{\otimes n}, \A^n) \ge \alpha-2\delta, \] by applying Pinsker's inequality.
\end{proof}
For the remainder of the section we let $\gamma = \exp(-1/\alpha^{C_m'})$ for sufficiently large constant $C_m'$ (depending on $m$) and let $\delta = \gamma^3/11000$. Assume without loss of generality that the maximum in~\cref{lemma:notuniform} is attained by the first expectation. Now we will apply our inverse theorem to randomly restrict $f_A$ so that it correlates with a product function over the uniform distribution on $\A^n$ (or we obtain a density increment by other means). For $A \subseteq \A^n$ let $\mu(A) = m^{-n}|A|$ denote the density of $A$ with respect to the uniform measure on $\A^n$.
\begin{lemma}
\label{lemma:uniform}
Assume that \eqref{eq:a37} holds. Then there is a subset $I \subseteq [n]$ with $|I| \le n - n^{1/3}$ and $u \in \A^n$ such that either:
\begin{itemize}
    \item $\mu(A_{I\to u}) \ge \alpha + \delta$, or
    \item $\mu(A_{I\to u}) \ge \alpha-10\delta/\gamma$ and there is a product function $P := P_1\dots P_n$, where $P_i\colon \Sigma\to\mathbb{D}$ for each $i$, such that \[ \E_{x \sim \A^{\bar{I}}}[(f_A)_{I \to u}(x)P(x)] \ge \gamma. \]
\end{itemize}
\begin{proof}
Applying \cref{thm:main} (recall that $\mu$ is pairwise-connected by item 2 of \cref{thm:3ap}) we know that there is some $\beta \ge \Omega_m\Big(\frac{\delta\gamma}{\sqrt{n}}\Big) \ge 2n^{-2/3}$ such that:
\[ \Pr_{\substack{I \sim_{1-\beta} [n] \\ u \sim \nu^I}}\Big[\sup\{|\l (f_A)_{I\to u}, P \r| : P \text{ is 1-bounded product function}\} \ge \gamma \Big] \ge \gamma,
\]
for some distribution $\nu$ satisfying $\mu = (1-\beta)\nu + \beta U$.
Indeed, we first pass via a random restriction from the distribution $\mu$ on $n$ coordinates to a distribution $\mu'$ on $\Theta(\delta\sqrt{n})$ coordinates, where the support of $\mu'$ is $S$, the probability of each atom in $\mu'$ is $\Omega_{m}(1)$, the marginal of $\mu'$ on $x$ is uniform, and
get the an expectation analogous to~\eqref{eq:a37} 
is still at least $\alpha^3/16$ with probability at least $\alpha^3/16$. We then invoke \cref{thm:main} 
on this expectation.

Thus item 2 holds unless
\begin{equation}
    \Pr_{\substack{I \sim_{1-\beta} [n] \\ u \sim \nu^I}}\Big[\mu(A_{I\to u}) \ge \alpha-10\delta/\gamma \Big] \le 1-0.99\gamma. \label{eq:goesdown}
\end{equation}
In that case, by \cref{lemma:notuniform}, \[ \E_{\substack{I \sim_{1-\beta} [n] \\ u \sim \nu^I}}[\mu(A_{I \to u})] = \E_{x \sim \mu^{\otimes n}}[A(x)] \ge \alpha - 2\delta. \] Combining this with \eqref{eq:goesdown} implies item 1.
\end{proof}
\end{lemma}
If item 1 of \cref{lemma:uniform} holds, then we are done (we have obtained a density increment). Otherwise, item 2 holds. In this case, for simplicity rename $\bar{I}$ as $[n]$ again (recall that $|\bar{I}| \ge n^{1/3}$ by \cref{lemma:uniform}), and rename $A_{I \to u}$ as $A$. Thus $(f_A)_{I \to u}$ gets renamed as $f_A = A - \alpha$.

In this case we require a more specialized kind of random restriction. We will try to find subsets $T \subseteq [n]$ and restrict to $x \in \A^n$ such that $x_t = x_{t'}$ for all $t, t' \in T$. We hope for two things to be true for this subset $T$. First, $\prod_{t \in T} P_t$ should be nearly-constant. Second, the density of $A$ does not drop significantly under this restriction. While this is clearly not true for a fixed $T$, we prove that we can choose $T$ with sufficient randomness to ensure this.

We now formally define the operation which forces coordinates to take the same value.
\begin{definition}
\label{def:same}
For a function $g: \A^n \to \bbC$ and $T \subseteq [n]$, define $g_{=T}: \A^{n-|T|+1} \to \bbC$ as follows. For $x \in \A$ and $y \in \A^{[n] \setminus T}$ let $v \in \A^n$ be the vector where $v_i = y_i$ for $i \in [n] \setminus T$ and $v_i = x$ otherwise. Then $g_{=T}(x,y) := g(v)$. For brevity, for disjoint sets $T_1, \dots, T_N$ we write $g_{=T_1,\dots,T_N} = (\dots(g_{=T_1})\dots)_{=T_N}$.
\end{definition}
We now prove that if $T$ is a sufficiently random set, then $\E[g_{=T}]$ is close to $\E[g]$.
\begin{lemma}
\label{lemma:same}
Let $S \subseteq [n]$ and let $g: \A^n \to \bbC$ be $1$-bounded. Then
\[ \left|\E_{x \sim \A^n} g(x) - \E_{T \subseteq S, |T| = k} \E_{x \sim \A^{n-k+1}} g_{=T}(x)\right| \le \frac{2mk}{\sqrt{|S|}}. \]
\end{lemma}
\begin{proof}
Let $\nu$ be the distribution on $\A^n$ over the $v$ vector (as defined in \cref{def:same}) obtained by sampling $T \subseteq S, |T| = k$ and $x \in \A^{n-k+1}$. It suffices to prove that $\dtv(\nu, \A^n) \le \frac{2mk}{\sqrt{|S|}}$ because $g$ is $1$-bounded. To prove this, further let $\nu'$ to be the distribution where we sample $T \subseteq S$ and set $x_t = a$ for all $t \in T$, where $a$ is some fixed element of $\A$. We will prove that $\dtv(\nu', \A^n) \le \frac{2mk}{\sqrt{|S|}}$.

Towards this, we may assume that $S = [n]$. The claim is evident when $k = 1$. For $k > 1$ note that sampling $T \subseteq [n]$ and be achieved by sampling $T' \subseteq [n]$ with $|T| = k-1$ and then sampling $t \in [n] \setminus T'$ and setting $T = T' \cup \{t\}$. Thus the result follows by induction.
\end{proof}
With these tools in hand we are ready to prove \cref{thm:3ap}.
\begin{proof}[Proof of \cref{thm:3ap}]
We may assume that item 2 of \cref{lemma:uniform} holds, or else we have already obtained a density increment. Let $v_1, \dots, v_n: \A \to \R/\Z$ be such that $P_j(x) = e^{2\pi i v_j(x)}$ for all $j \in [n]$, $x \in \A$.

Let $\zeta = \frac{1}{8m^2}$, and let $N = n^{\zeta}$. By the Pigeonhole principle, we can find disjoint sets $S_1, \dots, S_N$ of size $|S_i| = \sqrt{n}/2$ such that for every $i \in [N]$ and $j, j' \in S_i$, it holds that $\|v_j - v_{j'}\|_\infty \le n^{-\frac{1}{2m}}$. Let $v$ be an arbitrary vector in $S_i$, and let $1 \le k_i \le n^{\frac{1}{4m}}$ be such that $\|k_iv\|_\infty \le n^{-\frac{1}{4m^2}}$ -- such a $k_i$ exists by Dirichlet approximation.

Now let $T_i$ be a uniformly random subset of $S_i$ of size $k_i$ for all $i \in [N]$ and perform the following sequence of restrictions. For all $i \in [N]$ force that $x_j = x_{j'}$ for all $j, j' \in T_i$, and then uniformly random restrict all coordinates in $[n] \setminus (T_1 \cup \dots \cup T_N)$. By \cref{lemma:same} we know that:
\begin{equation} \E_{T_i \subseteq S_i, |T_i| = k_i, i \in [N]} \E_{u \sim \A^{[n] \setminus \cup_i T_i}}[\mu((A_{=T_1,\dots,T_N})_{I\to u})] \ge \alpha - 10\delta/\gamma - \frac{4Nmn^{\frac{1}{4m}}}{n^{1/4}} \ge \alpha - 11\delta/\gamma. \label{eq:final1} \end{equation}
Similarly, by \cref{lemma:same} we know that
\begin{equation} \E_{T_i \subseteq S_i, |T_i| = k_i, i \in [N]} \E_{u \sim \A^{[n] \setminus \cup_i T_i}} \E_{x \sim \A^N}\left[((f_A)_{=T_1,\dots,T_N})_{I \to u}(x) (P_{=T_1,\dots,T_N})_{I\to u}(x) \right] \ge \gamma -  \frac{4Nmn^{\frac{1}{4}}}{n^{1/4}} \ge \gamma/2. \label{eq:final2}
\end{equation}
Let $Q = (P_{=T_1,\dots,T_N})_{I\to u}$. We will argue that $Q$ is nearly constant. For $x, y \in \A^N$,
\begin{align*}
    Q(x) - Q(y) = \exp\Big(2\pi i \sum_{j=1}^N \sum_{t \in T_j} v_t(y)\Big)\left(\exp\Big(2\pi i \sum_{j=1}^N \sum_{t \in T_j} (v_t(x)-v_t(y))\Big) - 1 \right)\prod_{i \in ([n] \setminus \cup_i T_i)} P_i(u_i).
\end{align*}
For $j \in [N]$ and let $v \in T_j$ be such that $\|k_jv\|_\infty \le n^{-\frac{1}{4m^2}}$. Then
\[ \Big\|\sum_{t \in T_j} v_t \Big\|_\infty \le \|k_jv\|_\infty + \sum_{t \in T_j} \|v_t - v\|_\infty \le n^{-\frac{1}{4m^2}} + k_j n^{-\frac{1}{2m}} \le 2n^{-\frac{1}{4m^2}}. \]
Combining this with the above yields that $|Q(x) - Q(y)| \le 100N n^{-\frac{1}{4m^2}} = 100n^{-\frac{1}{8m^2}}.$ Letting $A' := (A_{=T_1,\dots,T_N})_{I\to u}$, combining this with the above \eqref{eq:final2} then implies that:
\[  \E_{\substack{T_i \subseteq S_i, |T_i| = k_i, i \in [N] \\ u \sim \A^{[n] \setminus \cup_i T_i}}}[|\mu(A')-\alpha|] \ge \gamma/2-100n^{-\frac{1}{8m^2}} \ge \gamma/3, \]
and thus
\[ \Pr_{\substack{T_i \subseteq S_i, |T_i| = k_i, i \in [N] \\ u \sim \A^{[n] \setminus \cup_i T_i}}}[|\mu(A')-\alpha| \ge \gamma/6] \ge \gamma/6. \]
If \[ \Pr_{\substack{T_i \subseteq S_i, |T_i| = k_i, i \in [N] \\ u \sim \A^{[n] \setminus \cup_i T_i}}}[\mu(A') > \alpha - \gamma/6] > 1-\gamma/6 \] then there is some event $A'$ where $\mu(A') > \alpha-\gamma/6$ while $|\mu(A')-\alpha| \ge \gamma/6$, so $\mu(A') \ge \alpha+\gamma/6$ as desired. Otherwise, we know that \[ \Pr_{\substack{T_i \subseteq S_i, |T_i| = k_i, i \in [N] \\ u \sim \A^{[n] \setminus \cup_i T_i}}}[\mu(A') \le \alpha - \gamma/6] \ge \gamma/6. \] Finally, recall that $\delta = \gamma^3/11000$. Combining this with \eqref{eq:final1}, which says that 
\[ \E_{\substack{T_i \subseteq S_i, |T_i| = k_i, i \in [N] \\ u \sim \A^{[n] \setminus \cup_i T_i}}}[\mu(A')] \ge \alpha - 11\delta/\gamma \ge \alpha - \gamma^2/1000 \] gives that there is some event $A'$ with $\mu(A') \ge \alpha + \gamma^2/72$.

The claimed quantitative dependence in \cref{thm:3ap} holds because $\gamma = \exp(-1/\alpha^C)$, and $A' \subseteq \A^N$ for $N \ge n^{\frac{1}{24m^2}}$ (at some point we labelled $n^{1/3}$ as $n$). Thus, after $O(1/\gamma)$ iterations the dimension would be at least 
$n^{(1/24m^2)^{O(1/\gamma)}}$ and we would end
up with a set with density $1-o(1)$, which clearly must contain a valid triplet.
\end{proof}

\subsection{Direct Sum Testing}\label{sec:dst}
In this section we prove~\cref{thm:direct_sum}, and we fix $f$ as in the statement of~\cref{thm:direct_sum} throughout.
We first note that the direct sum test is equivalent to the following test: sample $I\subseteq_{1/2}[n]$, 
$x,x'\sim \mathbb{F}_p^I$, 
$x,x'\sim \mathbb{F}_p^{\overline{I}}$ and 
check that 
$f(x,y) - f(x,y') - f(x',y) + f(x',y') = 0$. Thus,
it follows that
\[
\E\left[\sum\limits_{j\in\mathbb{F}_p}\omega_p^{j(f(x,y) - f(x,y') - f(x',y) + f(x',y'))}\right]
\geq \left(\frac{1}{p}+\eps\right)p
=1+p\eps,
\]
and so
$\E_x\left[\sum\limits_{j\neq 0}\omega_p^{j(f(x,y) - f(x,y') - f(x',y) + f(x,y))}\right]
\geq p\eps$. It follows there is $j\neq 0$ such that,
letting $F(z) = \omega_p^{j f(z)}$, we have that
\[
\left|\E_{}\left[F(x,y)\overline{F(x',y)}\overline{F(x,y')}F(x',y')\right]\right|\geq \eps.
\]
We note that the left hand side is exactly $\val(F)$, so applying~\cref{thm:swap,thm:boundedprod,thm:rit} we get that there is $L\colon \Sigma^n\to\mathbb{C}$ with $\norm{L}_2\leq 1$ and degree at most 
$D = \exp((1/\eps)^{O(1)})$, a product function $P = P_1\cdots P_n$ where $P_i\colon \Sigma\to\mathbb{C}$ is $1$-bounded, such that 
$|\langle F, L P\rangle|\geq \delta$, and $\delta=\exp(-\exp((1/\eps)^{O(1)}))$. Letting $G = \frac{1}{\norm{L}_4}F \overline{L}$, we get that
$\norm{G}_4\leq 1$ and 
\[
|\langle G, P\rangle|
=
\frac{|\langle F, L P\rangle|}{\norm{L}_4}
\geq \frac{\delta}{2^{O(D)}},
\]
where we used $\norm{L}_4\leq 2^{O(D)}$ which follows by hypercontractivity. Applying ~\cref{thm:boundedprod} again gives that there is a product function $Q = Q_1\cdots Q_n$ where 
$Q_i\colon \Sigma\to\mathbb{D}$ for each $i$, such that $|\langle G, Q\rangle|\geq \frac{\delta^{O(1)}}{2^{O(D)}}$, and plugging in the definition of $G$ finishes the proof.

\subsection{Proof of~\cref{thm:main_global}}
\cref{thm:main_global} follows by combining \cref{thm:main,thm:rit} and then applying~\cref{thm:boundedprod} as in~\cref{sec:dst}.

{\small
\bibliographystyle{alpha}
\bibliography{refs}}

\end{document}